\documentclass[12pt, draftclsnofoot, onecolumn]{IEEEtran}
\makeatletter
\newcommand\semihuge{\@setfontsize\semihuge{22.3}{22}}
\makeatother

\usepackage[dvips]{color}
\usepackage{comment}
\usepackage{todonotes}
\usepackage{epsf}
\usepackage{epsfig}
\usepackage{times}
\usepackage{epsfig}
\usepackage{graphicx}
\usepackage{bbold}
\usepackage{mathtools}
\usepackage{mathrsfs}
\usepackage{amssymb}
\usepackage{pdfpages}
\usepackage{epstopdf}
\usepackage{dsfont}
\usepackage{lettrine} 
\usepackage{amsmath,epsfig,amssymb,algorithm,algpseudocode,amsthm,cite,url}
\usepackage{subcaption}
\allowdisplaybreaks
\usepackage{csquotes}
\usepackage{color}
\usepackage{verbatim}
\usepackage[english]{babel}
\usepackage{amsmath,amssymb}

\usepackage[justification=centering]{caption}
\usepackage{verbatim}

\newtheorem{theorem}{\bf Theorem}

\newtheorem{proposition}{\bf Proposition}

\usepackage{xcolor}

\newenvironment{MyColorPar}[1]{%
	\leavevmode\color{#1}\ignorespaces%
}{%
}%

\begin{document}
	
\title{\semihuge Mobile Unmanned Aerial Vehicles (UAVs) for Energy-Efficient Internet of Things Communications}    
\author{\IEEEauthorblockN{ \normalsize Mohammad Mozaffari$^1$, Walid Saad$^1$, Mehdi Bennis$^2$, and M\'erouane Debbah$^3$}\vspace{-0.05cm}\\
	\IEEEauthorblockA{
		\small $^1$ Wireless@VT, Electrical and Computer Engineering Department, Virginia Tech, VA, USA,\\ Emails:\url{{mmozaff , walids}@vt.edu}.\\
		$^2$ CWC - Centre for Wireless Communications, Oulu, Finland, Email: \url{bennis@ee.oulu.fi}.\\
		$^3$ Mathematical and Algorithmic Sciences Lab, Huawei France R \& D, Paris, France, and CentraleSup´elec,\\   Universit\'e Paris-Saclay, Gif-sur-Yvette, France, Email: \url{merouane.debbah@huawei.com}.
	}\vspace{-0.9cm}}
\maketitle\vspace{-0.7cm}
\vspace{-0.1cm}
\begin{abstract}\vspace{-0.2cm}
In this paper, the efficient deployment and mobility of multiple unmanned aerial vehicles (UAVs), used as aerial base stations to collect data from ground Internet of Things (IoT) devices, is investigated. In particular, to enable reliable uplink communications for IoT devices with a minimum total transmit power, a novel framework is proposed for jointly optimizing the three-dimensional (3D) placement and mobility of the UAVs, device-UAV association, and uplink power control. First, given the locations of active IoT devices at each time instant, the optimal UAVs' locations and associations are determined. \textcolor{black}{Next, to dynamically serve the IoT devices in a time-varying network, the optimal mobility patterns of the UAVs are analyzed. To this end, based on the activation process of the IoT devices, the time instances at which the UAVs must update their locations are derived. Moreover, the optimal 3D trajectory of each UAV is obtained in a way that the total energy used for the mobility of the UAVs is minimized while serving the IoT devices. Simulation results show that, using the proposed approach, the total transmit power of the IoT devices is reduced by 45\% compared to a case in which stationary aerial base stations are deployed. In addition, the proposed approach can yield a maximum of 28\% enhanced system reliability compared to the stationary case.  The results also reveal}  an inherent tradeoff between the number of update times, the mobility of the UAVs, and the transmit power of the IoT devices. In essence, a higher number of updates can lead to lower transmit powers for the IoT devices at the cost of an increased mobility for the UAVs.\vspace{-0.45cm}

 
\end{abstract} \vspace{-0.04cm}
\section{Introduction}\vspace{-0.1cm}

The use of unmanned aerial vehicles (UAVs) as flying wireless communication platforms has received significant attention recently  \cite{mozaffari2, zhang, Jiang,zengThroughput, MozaffariIoT,Irem, LetterNew, HouraniOptimal}. 
On the one hand, UAVs can be used as wireless relays for improving connectivity and coverage of ground wireless devices. On the other hand, UAVs can act as mobile aerial base stations to provide reliable downlink and uplink communications for ground users and boost the capacity of wireless networks \cite{HouraniOptimal, Irem, mozaffari2,Letter, zhang,VishnuJ,LetterNew,Mingzhe, Hover}. Compared to the terrestrial base stations, the advantage of using UAV-based aerial base stations is their ability to provide on-the-fly communications. \textcolor{black}{Furthermore, the high altitude of UAVs enables them to effectively establish line-of-sight (LoS) communication links thus mitigating signal blockage and shadowing.} 
Due to their adjustable altitude and mobility, UAVs can move towards potential ground users and establish reliable connections with a low transmit power. Hence, they can provide a cost-effective and energy-efficient solution to collect data from ground mobile users that are spread over a geographical area with limited terrestrial infrastructure. 

Indeed, UAVs can play a key role in the \emph{Internet of Things (IoT)} which is composed of small, battery-limited devices such as  
sensors, and health monitors \cite{lien}. These devices are typically unable to transmit over a long distance due to their energy constraints \cite{lien}. In such IoT scenarios, UAVs can dynamically move towards IoT devices, collect the IoT data, and transmit it to other devices which are out of the communication ranges of the transmitters \cite{lien}. In this case, the UAVs play the role of moving aggregators or base stations for IoT networks \cite{MozaffariIoT}. However, to effectively use UAVs for the IoT, several challenges must be addressed such as optimal deployment, mobility and energy-efficient use of UAVs as outlined in \cite{Irem} and \cite{mozaffari2}.

\textcolor{black}{In \cite{Jiang}, the authors investigated the optimal trajectory of UAVs equipped with multiple antennas for maximizing sum-rate in uplink communications. The work in \cite{zengThroughput} maximizes the throughput of a relay-based UAV system by jointly optimizing the UAV's trajectory as well as the source/relay transmit power.} However these works considered a single UAV in their models. In \cite {mozaffari2}, we investigated the optimal deployment and movement of a single UAV for supporting downlink wireless communications.  
\textcolor{black}{The work in \cite{LyuPlacement} proposed a low-complexity algorithm for the optimal deployment of multiple UAVs that provide coverage for ground users.}
The work in \cite{VishnuJ} provided  a comprehensive downlink coverage analysis for a network in which a finite number of UAVs serve the ground users. 
 In \cite{pang}, the authors used UAVs to efficiently collect data and recharge the clusters' head in a wireless sensor network which is partitioned into multiple clusters. However, this work is limited to a static sensor network, and does not investigate the optimal deployment of the UAVs. While the energy efficiency of uplink data transmission in a machine-to-machine (M2M) communication network was investigated \cite{Tu}, the presence of UAVs was not considered. 
In fact, none of the prior studies in \cite{MozaffariIoT,Irem, Jiang,zengThroughput, mozaffari2, lien, zhang, HouraniOptimal, Letter, pang, Tu,VishnuJ,LyuPlacement,LetterNew,Mingzhe,Hover}, addressed the problem of jointly optimizing the deployment and mobility of UAVs, device association, and uplink power control for enabling reliable and  energy-efficient communications for IoT devices. To our best knowledge, this paper is \textcolor{black}{\emph{\textcolor{black}{one of the first} comprehensive studies on the joint optimal 3D deployment of aerial base stations, device association, and uplink power control in an IoT ecosystem}.} 



The main contribution of this paper is to introduce a novel framework for optimized deployment and mobility of multiple UAVs for the purpose of energy-efficient uplink data collection from ground IoT devices. In particular, we consider an IoT network in which the IoT devices can be active at different time instances. To minimize the total transmit power of these IoT devices, given device-specific signal-to-interference-plus-noise-ratio (SINR) constraints, we propose an efficient approach to jointly and dynamically find the UAVs' locations, the association of devices to UAVs, and the optimal uplink transmit power. Our proposed framework is composed of two key steps. First, given the locations of the IoT devices, we propose a solution for optimizing the deployment and association of the UAVs. In this case, we solve the formulated problem by decomposing it into two subproblems which are solved iteratively. In the first subproblem, given the fixed UAVs' locations, we find the jointly optimal device-UAV association and the devices' transmit power. In the second subproblem, given the fixed device association, we determine the joint 3D UAVs' locations. For this subproblem, we transform the non-convex continuous location optimization problem to a convex form and provide tractable solutions. Next, following our proposed algorithm, the results of solving the second subproblem are used as inputs to the first subproblem for the next iteration. Here, we show that our proposed approach leads to an efficient solution with a reasonable accuracy compared to the global optimal solution that requires significant overhead. Clearly, the UAVs' locations and the device association that we obtain in this first step will depend on the locations of active IoT devices.    

In the second step, we analyze the IoT network over a time period during which the set of active devices changes. In this case, we present a framework for optimizing the UAVs' mobility by allowing them dynamically update their locations depending on the time-varying devices' activation process. First, we derive the closed-form expressions for the time instances (update times) at which the UAVs must move according to the activation process of the devices. Next, using the update time results, we derive the optimal 3D UAVs' trajectory such that the total movement of the UAVs while updating their locations is minimized. \textcolor{black}{Our simulation results show that, using the proposed approach, the total transmit power of the IoT devices can be significantly reduced compared to a case in which stationary aerial base stations are deployed.}  The results also verify our analytical derivations for the update times and reveal an inherent tradeoff between the number of updates, the mobility of the UAVs, and transmit power of the IoT devices. In particular, it is shown that a higher number of updates leads to lower transmit powers for the IoT devices at the cost of higher UAVs' energy consumptions.    


The rest of this paper is organized as follows. In Section II, we present the system model and problem formulation. Section III presents the optimal deployment of UAVs and device association. In Section IV, we address the mobility and update time of the UAVs. In Section V we provide the simulation and analytical results, and Section VI draws some conclusions.\vspace{-0.2cm}

\section{System Model and Problem Formulation}

Consider an IoT system consisting of a set $\mathcal{L}=\{1,2,...,L\}$ of $L$ IoT devices. Examples of such devices include various types of sensors used for environmental monitoring, smart traffic control, and  smart parking devices. In this system, a set $\mathcal{K}=\{1,2,...,K\}$ of $K$ \textcolor{black}{rotary wing UAVs} must be deployed to collect the data from the ground IoT devices. These UAVs can dynamically move, when needed, to effectively serve the IoT devices using uplink communication links. \textcolor{black}{Here, the term \textit{served} by a UAV implies that the uplink SINR is above the threshold and, thus, the UAV can successfully collect data from the ground IoT device.} In our model, we assume that the devices transmit their data to the UAVs in the uplink using frequency division multiple access (FDMA) over $R$ orthogonal channels. Let $E_\textrm{max}$ be the maximum energy that each UAV can spend on its movement. The locations of device $i\in \mathcal{L}$ and UAV $j\in \mathcal{K}$ are, respectively, given by $(x_i,y_i)$ and $\boldsymbol{v}_j=(x^\textrm{uav} _{j},y^\textrm{uav} _{j},h_j)$ as shown in Fig. \ref{System}. In our model, we consider a centralized network in which the locations of the devices and UAVs are known to a control center located at a
central \emph{cloud} server. The cloud server will determine the UAVs' locations, device association, and the transmit power of each IoT device.

We analyze the IoT network within a time interval $[0,T]$ during which the IoT devices can be active at different time instances and must be served by the UAVs at some pre-defined time slots. At the beginning of each slot, the positions of the UAVs as well as the device-UAV association are updated based on the locations of currently active devices \textcolor{black}{that are assumed to be known to the cloud center\footnote {We consider static IoT devices in delay-tolerant applications which their fixed locations and activation patterns are known to the cloud center.}.} Hereinafter, the time instance at which the UAVs' locations and associations are \emph{jointly} updated, is referred to as the \emph{update time}. The update times are denoted by $t_n$, $1\leq n\leq N$, with $N$ being the number of updates.  
At each update time $t_n$, based on the location of active devices, the optimal UAVs' locations and the corresponding association must be determined for effectively serving the ground devices. Here, the IoT devices that become active during $[{t_{n-1}},{t_{n}})$ are served by the UAVs during the time period $[{t_{n}},{t_{n+1}})$. Note that, during $[{t_{n-1}},{t_{n}})$, the UAVs' locations and their device association do not change until the next update time, $t_{n}$.  
Clearly, since at different update times, a different subset of devices might be active, the locations of the UAVs must dynamically change at each update time. Therefore, each UAV's \emph{trajectory} will consist of $N$ stop locations at which the UAV serves the ground devices. \textcolor{black}{Note that, in our
	model, the UAVs’ locations are not necessarily updated once the set of active devices changes.
	Instead, we consider some specific time instances (update times) at which the UAVs locations
	device associations, and devices' transmit power are optimized. In particular, considering the fact that
	the set of active devices may continuously change, continuously updating the UAVs' locations,
	the devices transmit powers, and the device-UAV associations may not be feasible as it can lead to low reliability, high UAVs' energy consumption, and a need to solve complex real-time optimization processes. In
	our model, the update times are design parameters that depend on the activity of the devices,
	and the energy of UAVs.} Given this model, our objective is to find the optimal joint UAVs' locations and device association at each update time $t_n$ so as to minimize the total transmit power of the active devices while meeting each device's SINR requirement. Moreover, we need to develop a framework for determining the update times as well as the UAVs' mobility to handle dynamic changes in the activity of the devices. 
 To this end, first, we present the ground-to-air channel model and the activation models for the IoT devices. \vspace{-0.2cm}
 
 \begin{figure}[!t]
 	\begin{center}
 		\vspace{-0.2cm}
 		\includegraphics[width=8.3cm]{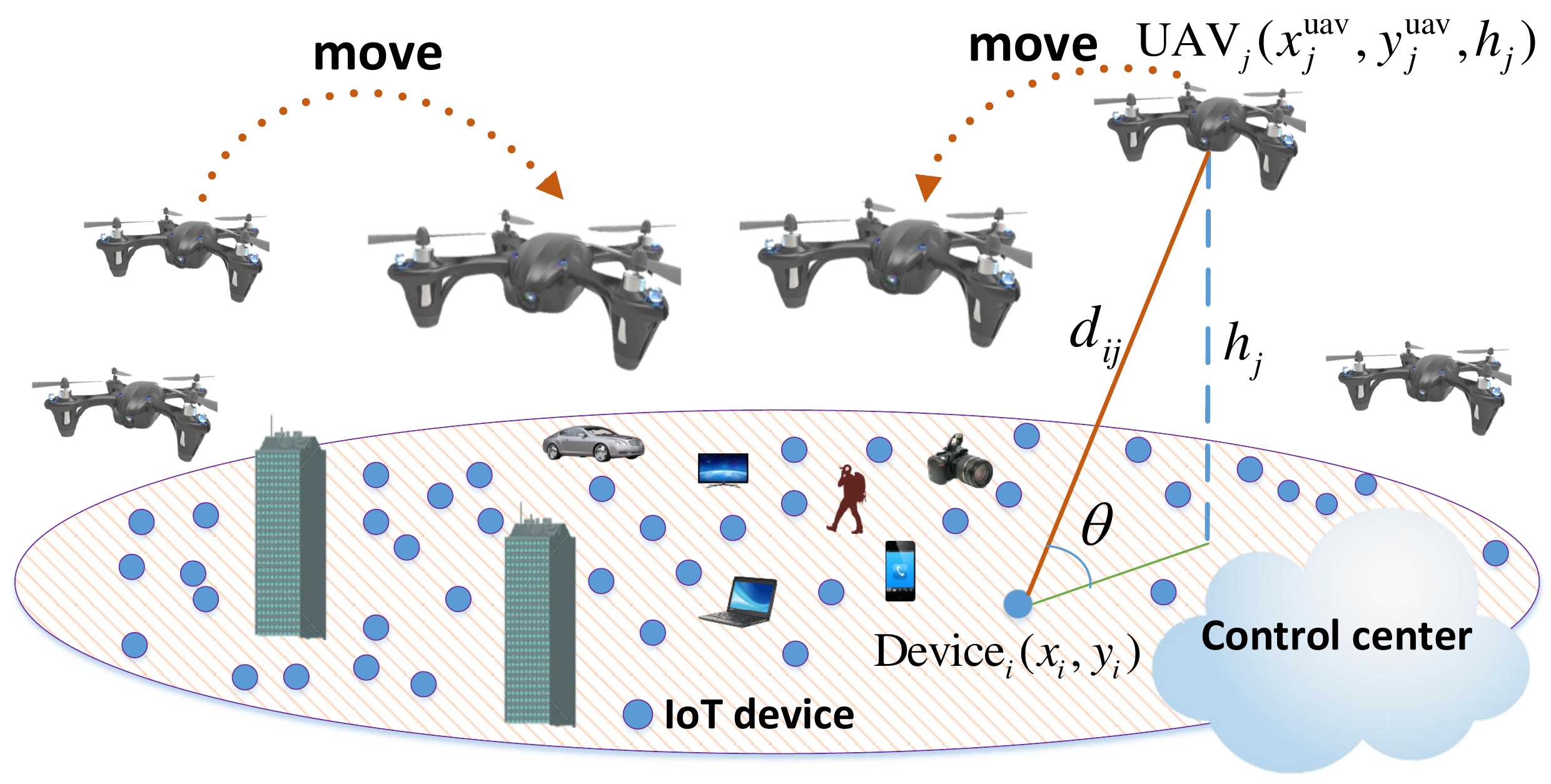}
 		\vspace{-0.9cm}
 		\caption{ \small System model.\vspace{-1.2cm}}
 		\label{System}
 	\end{center}
 \end{figure}  
 
 \subsection{Ground-to-Air Path Loss Model}

\textcolor{black}{In our model, while optimizing the locations of the UAVs, the information available includes the ground devices' locations, and the type of environment (e.g. rural, suburban, urban, highrise urban, etc.). Note that, in such practical scenarios, one will not have any additional information about the exact locations, heights, and number of the obstacles. Therefore, one must consider the randomness associated with the LoS and NLoS links while designing the UAV-based communication system.}
 Therefore, for  ground-to-air communications, each device will typically have a LoS view towards a specific UAV with a given probability. This LoS probability depends on the environment, location of the device and the UAV as well as the elevation angle \cite{HouraniOptimal}. One suitable expression for the LoS probability is given by \cite{Irem,mozaffari2,HouraniOptimal}:\vspace{-0.09cm}
\begin{equation}\label{PLoS}
	{P^{ij}_{{\rm{LoS}}}} = \frac{1}{{1 + \psi \exp ( - \beta\left[ {\theta_{ij}  - \psi} \right])}},
\end{equation}
where $\psi$  and $\beta$  are constant values which depend on the carrier frequency and type of environment such as rural, urban, or dense urban, and $\theta_{ij}$  is the elevation angle. Clearly, ${\theta} = \frac{{180}}{\pi } \times {\sin ^{ - 1}}\left( {{\textstyle{{{h_j}} \over { {d_{ij}}}}}} \right)$, where $ {d_{ij}} = \sqrt {(x_i-x^\textrm{uav}_j)^2+(y_i-y^\textrm{uav}_j)^2+h_j^2 }$ is the distance between device $i$ and UAV $j$. 

From (\ref{PLoS}), we can see that by increasing the elevation angle or increasing the UAV altitude, the LoS probability increases. The path loss model for LoS and  non-line-of-sight (NLoS) links between device $i$ and UAV $j$ is given by \cite{Irem} and \cite{HouraniOptimal}:
\begin{equation}\label{pathloss}
{L_{ij}} = \left\{ \begin{array}{l}
{\eta_1 \left( {\frac{{4\pi {f_c}{d_{ij}}}}{c}} \right)^\alpha},{\rm{{\rm \hspace*{1.3cm}{LoS\hspace*{.2cm} link}}}},\\
\eta_2 {\left( {\frac{{4\pi {f_c}{d_{ij}}}}{c}} \right)^\alpha },{\rm \hspace*{1.3cm} {   NLoS \hspace*{.2cm}link}},
\end{array} \right.
\end{equation}
where $f_c$ is the carrier frequency, $\alpha$ is the path loss exponent, $\eta_1$ and $\eta_2$ ($\eta_2>\eta_1>1$) are the excessive path loss coefficients in LoS and NLoS cases, and $c$ is the speed of light. Note that, the NLoS probability is $P^{ij}_{\text{NLoS}}=1-P^{ij}_{\text{LoS}}$.
Typically, given only the locations of the UAVs and devices, it is not possible to exactly determine which path loss type (LoS/NLoS) is experienced by the device-UAV link. In this case, the path loss average considering both LoS and NLoS links can be used for the device-UAV communications \cite{Irem} and \cite{HouraniOptimal}. Now, using (\ref{PLoS}) and (\ref{pathloss}), the average path loss between device $i$ and UAV $j$ can be expressed as:
\begin{equation}\label{L_ave}
\bar L_{ij} = {P^{ij}_{{\rm{LoS}}}}\eta_1 \left( {\frac{{4\pi {f_c}{d_{ij}}}}{c}} \right)^{\alpha} +  {P^{ij}_{{\rm{NLoS}}}}\eta_2 \left( {\frac{{4\pi {f_c}{d_{ij}}}}{c}}\right)^{\alpha}=\left[{P^{ij}_{{\rm{LoS}}}}\eta_1+{P^{ij}_{{\rm{NLoS}}}}\eta_2\right] \left( K_o d_{ij} \right)^{\alpha},
\end{equation} 
where $K_o=\frac{{4\pi {f_c}}}{c}$. Clearly, the average channel gain between the UAV and the device is $\bar g_{ij}=~\frac{1}{\bar L_{ij}}$. \textcolor{black}{ Note that, by using the average channel gain,  there is no need to account for LoS and NLoS links separately, and, hence, the SINR expressions become more tractable. Therefore, we use the average channel gain to model the interference and desired links for all device-UAV communications while computing the SINRs.\vspace{-0.4cm}}

 \subsection{IoT Device Activation Model}
Indeed, the activation of IoT devices depends on the services that they are supporting. For instance, in some applications such as weather monitoring, smart grids, or home automation, the IoT devices need to report their data periodically. However, the IoT devices can have random activations in health monitoring, or smart traffic control applications. Therefore, the UAVs must be properly deployed to collect the IoT devices data while dynamically adapting to the activity patterns of these devices. 
Naturally, the optimal locations of the UAVs and their update times depend on the activation process of the IoT devices. Here, we consider two activation models. In the first model, the IoT devices are randomly activated, as in smart traffic control applications. \textcolor{black}{In this case, the concurrent transmissions of a massive number of devices within a short time duration can lead to a bursty traffic as pointed out in \cite{chen} and \cite{Tavana}. In fact, when massive IoT devices attempt to transmit within a short time period, the arrival patterns become more bursty \cite{Jian}. Thus, 3GPP suggests a \textit{beta} distribution to capture this traffic characteristic of IoT devices \cite{3GPP}. In this case, each IoT device will be active at time $t\in[0,T]$ following the beta distribution with parameters $\kappa$ and $\omega$ \cite{3GPP,Jian,Tavana}:\vspace{-0.2cm}}
\begin{equation}\label{Beta}
f(t) = \frac{{{t^{\kappa  - 1}}{{(T - t)}^{\omega  - 1}}}}{{{T^{\kappa   + \omega  - 1}}B(\kappa  ,\omega )}},
\end{equation}
where $[0,T]$ is the time interval within which the IoT devices can be active, and $B(\kappa ,\omega ) = \int_0^1 {{t^{\kappa  - 1}}{{(1 - t)}^{\omega - 1}}{\rm{d}}t}$ is the beta function with parameters $\kappa$ and $\omega$ \cite{gupta}.

In addition, IoT devices such as smart meters typically report their data periodically rather than randomly. 
For such devices, the activation process is deterministic and assumed to be known in advance. In such case, we assume that device $i$ becomes active each $\tau_i$ seconds during $[0,T]$ time duration. Clearly, the number of activations for a device $i$ during $[0,T]$ is $\left\lfloor {\frac{T}{{{\tau _i}}}} \right\rfloor$.\vspace{-0.4cm} 
\textcolor{black}{
\subsection{Channel Assignment Strategy} 
Here, given only the devices' locations, a practical channel assignment approach is to assign different channels to devices which are located in proximity of each other. This approach significantly mitigates the possibility of having strong interference between two closely located devices. For the channel assignment problem, we have adopted a constrained \textit{K-mean clustering} strategy \cite{selim}, which is an efficient distance-based clustering approach in which a set of given points are grouped into $K$ clusters based on their proximity. In this case, given the number of active devices, $L_n$, and the number of  orthogonal channels, $R\le L_n$, we group the devices based on proximity, and assign different channels to devices which are in the same~ group.}

Now, we present our optimization problem to find the UAVs' locations, device association, and transmit power of IoT devices at each update time $t_n$ during $[0,T]$: 

\hspace{3cm}\textbf{(OP)}:\vspace{-0.4cm}
\begin{align}
&\mathop {\min}\limits_{\boldsymbol{v}_j, \boldsymbol{c},\boldsymbol{P}} \sum\limits_{i=1}^{{L_n}} {{P_i}}\, ,\,\,\,\,\,\,\, \forall i \in \mathcal{L}_n, \,\, \forall  j \in \mathcal{K},\label{Assosiation1}\\
\text{s.t.} \,\, &\frac{{{P_i}{{\bar g}_{ic_i}}(\boldsymbol{v}_{c_i})}}{{\sum\limits_{k \in \mathcal{Z}_i} {{P_k}{\bar g}_{kc_i}(\boldsymbol{v}_{c_i})}  + {\sigma ^2}}} \ge \gamma, \label{SINR}\\
&0 < {P_i} \le {P_\textrm{max}},\vspace{-0.3cm}\label{Pmax}
\end{align}
where $L_n$ is the total number of active devices at update time $t_n$, and $\mathcal{L}_n$ is the set of devices' index. $\boldsymbol{P}$ is the transmit power vector with each element $P_i$ being the transmit power of device $i$. Also, $\boldsymbol{v}_j$ is the 3D location of UAV $j$, and $\boldsymbol{c}$ is the device association vector with each element $c_i$ being the index of the UAV that is assigned to device $i$. $P_\textrm{max}$ is the maximum transmit power of each IoT device, and $\sigma ^2$ is the noise power. Furthermore, ${\bar g}_{ic_i}(\boldsymbol{v}_{c_i})$ is the average channel gain between device $i$ and UAV $c_i$ which is a function of the UAV's location. Also, ${\bar g}_{kc_i}(\boldsymbol{v}_{c_i})$ is the average channel gain between interfering device $k$ and UAV $c_i$. 
 In (\ref{SINR}), $\mathcal{Z}_i$ is the set of all other devices that use the same channel as device $i$ and create interference. $\gamma$ is the SINR target which must be achieved by all the devices, (\ref{SINR}) represents the SINR requirement, and (\ref{Pmax}) shows the maximum transmit power constraint. Hereinafter, we call \textbf{(OP)} the \textit{original problem}.

Note that, in (\ref{Assosiation1}), the transmit power of the IoT devices, the 3D locations of the UAVs, and the UAV-device associations are unknowns. Clearly, the locations of the UAVs impact the channel gain between the devices and UAVs, and, hence, they affect the transmit power of each device, $P_i$. Furthermore, given (\ref{SINR}), due to the mutual interference between the devices, the transmit power of each device depends also on the  transmit power of the interfering devices as well as the device-UAV associations. In addition, the device-UAV associations depend on the UAVs' locations which are also unknowns. Therefore, there is a mutual dependency between all the optimization variables in \textbf{(OP)}. Moreover, considering (\ref{PLoS}) and constraint (\ref{SINR}), we can see that, this optimization problem is highly non-linear and non-convex. Indeed, solving (\ref{Assosiation1}) is significantly  challenging due to the mutual dependency of the optimization variables, non-linearity, and non-convexity of the problem. Next, we propose a framework for solving this optimization problem. 

In essence, our proposed framework for solving \textbf{(OP)} proceeds as follows. At each update time $t_n$, given the fixed UAVs' locations, we find the optimal device-UAV association and the transmit power of the devices. Next, given the fixed UAV association from the previous step, we determine the \textcolor{black}{sub-optimal} locations of the UAVs and update the transmit power the devices accordingly. This procedure is done iteratively until the 3D UAVs' locations, device association, and the transmit power of the devices are found. Clearly, at each step, the total transmit power of the devices decreases, and, hence, the proposed algorithm converges. Fig. \ref{Flowchart} shows a block diagram that summarizes the main steps for solving \textbf{(OP)}.  Next, we discuss, in detail, each block of the proposed solution in Fig. \ref{Flowchart}. \vspace{-0.3cm}
\begin{figure}[!t]
	\begin{center}
		\vspace{-0.2cm}
		\includegraphics[width=15cm]{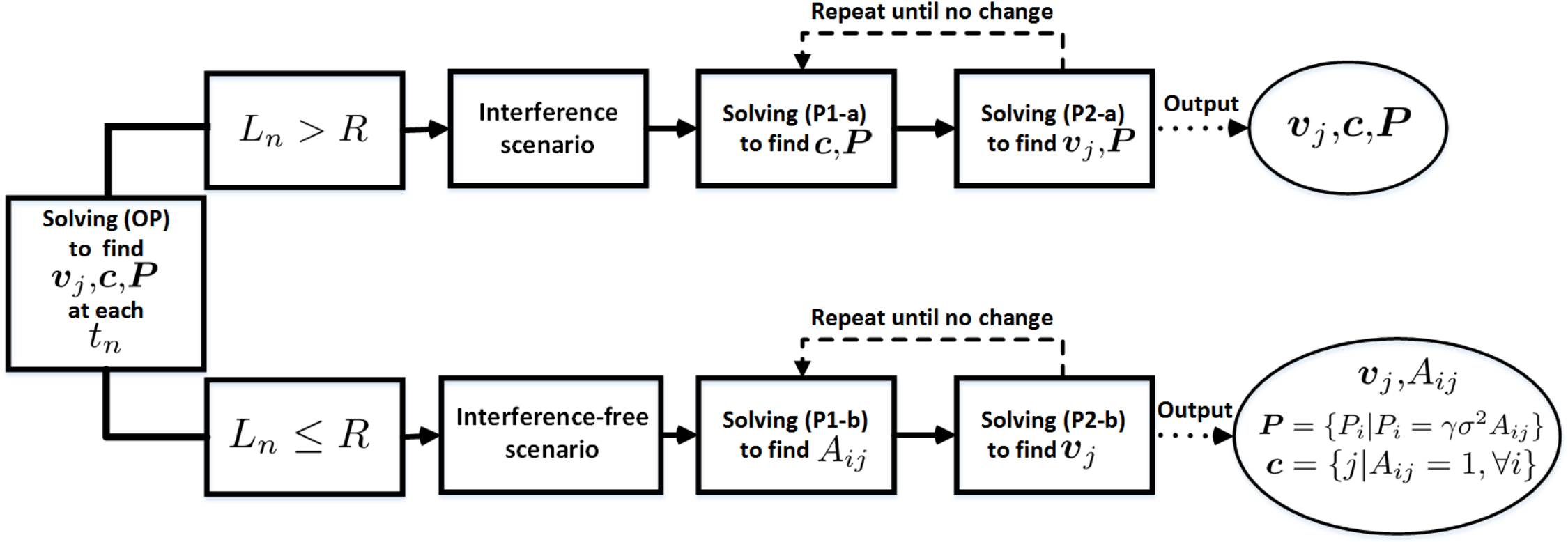}
		\vspace{-0.7cm}
		\caption{ \small Block diagram for the proposed solution.}\vspace{-1.4cm}
		\label{Flowchart}
	\end{center}
\end{figure}


\section{UAV Deployment and Device Association with Power control}
Here, given the locations of active IoT devices, we minimize the total transmit power of the devices by solving (\ref{Assosiation1}). Clearly, the UAVs' locations and the device association are mutually dependent. In particular, to find the device association, the locations of the UAVs must be known. Moreover, the UAVs' locations cannot be optimized without knowing the device association.  

Therefore, we decompose \textbf{(OP)} into two subproblems that will be solved iteratively. In the first subproblem, given the locations of the UAVs, we find the optimal device association and the transmit power of the devices such that the uplink SINR requirements of all active devices are satisfied with a minimum total transmit power. In the second subproblem, given the device association resulting from the first subproblem, we determine the \textcolor{black}{sub-optimal} locations of the UAVs for which the transmit power of the devices is minimized. Note that, this is an iterative process in which the results of each subproblem are used in the other subproblem for the next iteration. These computations are performed by the control center until the 3D UAVs' locations, device association, and transmit power of the devices are obtained. 
 
Note that, given the limited number of available orthogonal channels, the interference between the devices will depend on the number of active devices at each update time. Clearly, there is no interference when the number of active devices at time $t_n$ is less than the number of orthogonal channels, or equivalently $L_n\le R$. Given that, in the interference-free scenario, one can provide a more tractable analysis, here, for the deployment and association steps. Therefore, we will investigate the interference and interference-free scenarios, separately.\vspace{-0.4cm}
\subsection{Device Association and Power Control} \label{assosiationPowerControl}
Here, given initial locations of the UAVs, we aim to find the optimal device association as well as the transmit power of each IoT device such that the total transmit power used for successful uplink communications is minimized. \vspace{-0.001cm}
\subsubsection{Interference scenario}
In the presence of uplink interference \textcolor{black}{when $L_n>R$,} the power minimization problem at update time $t_n$ will be given by:\vspace{-0.1cm}

\hspace{3cm}\textbf{(P1-a)}:\vspace{-0.5cm}
\begin{align}
&\mathop {\min}\limits_{ \boldsymbol{c},\boldsymbol{P}} \sum\limits_{i=1}^{{L_n}} {{P_i}}\, ,\,\,\,\,\,\,\, \forall i \in \mathcal{L}_n, \,\, \forall j \in \mathcal{K},\label{AssosiationP1a}\\
\text{s.t.} \,\, &\frac{{{P_i}{{\bar g}_{ic_i}}}}{{\sum\limits_{k \in \mathcal{Z}_i} {{P_k}{\bar g}_{kc_i}}  + {\sigma ^2}}} \ge \gamma, \label{SINRP1a}\\
&0 < {P_i} \le {P_\textrm{max}}.\label{PmaxP1a}
\end{align}


To solve \textbf{(P1-a)}, we need to jointly find the optimal device association and the transmit power of all active devices under the SINR constraints for the given UAVs' locations. Clearly, given the fixed UAVs' locations, optimization variables are the device association and the transmit power of the devices. Note that, satisfying the SINR requirement of each device significantly depends on the distance and altitude of its serving UAV. Therefore, the feasibility of the optimization problem in (\ref{AssosiationP1a}) depends on the locations of the UAVs. Next, we derive an upper bound and a lower bound for the altitude of serving UAV $j$ as a function of its distance from device $i$. \vspace{-0.1cm}
\begin{proposition}
	\textnormal{
	The lower and upper bounds for the altitude of a UAV $j$ needed to serve a device $i$ (meeting its SINR requirement), are given by:  
\begin{equation}\label{h}
	{d_{ij}}\sin \left( {\frac{1}{\beta }\ln \left( {\frac{{\psi Q}}{{1 - Q}}} \right) + \psi } \right) \le {h_j} \le {\left( {\frac{{{P_\textrm{max}}}}{{\gamma K_o^\alpha {\sigma ^2}{\eta _1}}}} \right)^{1/\alpha }},
\end{equation}	
where $d_{ij}$ is the distance between UAV $j$ and device $i$, and $Q = \frac{{{P_\textrm{max}}}}{{\gamma d_{ij}^\alpha {K_{o}}^\alpha {\sigma ^2}\left( {{\eta _1} - {\eta _2}} \right)}} - \frac{{{\eta _2}}}{{{\eta _1} - {\eta _2}}}$.}
\end{proposition}
\begin{proof}
	Let $I_i$ be the cumulative interference from interfering devices on device $i$, then:
	\begin{align}
	&\textrm{SINR}_i = \frac{{{P_i}{{\bar g}_{ij}}}}{{{I_i} + {\sigma ^2}}} \ge \gamma, \nonumber\\
	&d_{ij}^\alpha  \le \frac{{{P_i}}}{{\gamma K_o^\alpha \left( {{I_i} + {\sigma ^2}} \right)\left( {{\eta _1}P_{\textrm{LoS}}^{ij} + {\eta _2}P_{\textrm{NLoS}}^{ij}} \right)}} \le \frac{{{P_\textrm{max}}}}{{\gamma K_o^\alpha {\sigma ^2}\left( {{\eta _1}P_\textrm{LoS}^{ij} + {\eta _2}(1 - P_{\textrm{LoS}}^{ij})} \right)}},\nonumber\\
	&P_\textrm{LoS}^{ij} \ge \frac{{{P_\textrm{max}}}}{{\gamma d_{ij}^\alpha K_o^\alpha {\sigma ^2}\left( {{\eta _1} - {\eta _2}} \right)}} - \frac{{{\eta _2}}}{{{\eta _1} - {\eta _2}}},\nonumber\\
	&\textrm{considering} \,\,Q = \frac{{{P_\textrm{max}}}}{{\gamma d_{ij}^\alpha K_o^\alpha {\sigma ^2}\left( {{\eta _1} - {\eta _2}} \right)}} - \frac{{{\eta _2}}}{{{\eta _1} - {\eta _2}}}, \textrm{and using equation (\ref{PLoS})},\nonumber\\
	&{\theta _{ij}}\mathop  \ge \limits^{(a)} \frac{1}{\beta }\ln \left( {\frac{{\psi Q}}{{1 - Q}}} \right) + \psi,\nonumber\\
	&{h_j} \ge {d_{ij}}\sin \left( {\frac{1}{\beta }\ln \left( {\frac{{\psi Q}}{{1 - Q}}} \right) + \psi } \right),\label{lower}
	\end{align}
	where $(a)$ stems from (\ref{PLoS}). Also, we have:
	\begin{equation}
	 d_{ij}^\alpha  \le \frac{{{P_\textrm{max}}}}{{\gamma K_o^\alpha {\sigma ^2}\left( {{\eta _1}P_\textrm{LoS}^{ij} + {\eta _2}(1 - P_\textrm{LoS}^{ij})} \right)}}\mathop  \le \limits^{(b)} \frac{{{P_\textrm{max}}}}{{\gamma K_o^\alpha {\sigma ^2}{\eta _1}}},
	 \end{equation}
	where in $(b)$ we consider $P_\textrm{LoS}=1$ which is equivalent to ${h_j} = {d_{ij}}$. Finally,   
	\begin{equation} \label{upper}
	{h_j} \le {\left( {\frac{{{P_\textrm{max}}}}{{\gamma K_o^\alpha {\sigma ^2}{\eta _1}}}} \right)^{1/\alpha }}.
	 \end{equation}
	 Clearly, (\ref{lower}) and (\ref{upper}) prove the proposition.
	 \end{proof}
Proposition 1 provides the necessary conditions for the UAV's altitude needed in order to be able to serve the IoT device. From (\ref{h}), the minimum altitude must increase as the distance increases. In other words, the UAV's altitude needs to be adjusted based on the distance such that the elevation angle between the device and the UAV exceeds $\frac{1}{\beta }\ln \left( {\frac{{\psi Q}}{{1 - Q}}} \right) + \psi$. Furthermore, as expected, the maximum altitude of the UAVs significantly depends on the maximum transmit power of the devices as given in (\ref{upper}).

Now, given the fixed UAVs' locations, problem \textbf{(P1-a)} corresponds to the problem of  joint user association and uplink power control in the terrestrial base station scenario. The algorithm presented in \cite{Yates} and \cite{Sun} leads to the global optimal solution to the joint user association and uplink power control under the SINR and maximum transmit power constraints. As a result, the optimal transmit power of users and the base station association for which the total uplink transmit power is globally minimized, is determined. In problem \textbf{(P1-a)}, the IoT devices correspond to the users, and fixed positioned UAVs correspond to the terrestrial base stations. For our case, this algorithm, as given in Algorithm\,1, will proceed as follows. \begin{MyColorPar}{black} We start with an initial value for transmit power of all active devices in step \ref{0}. Then, in step \ref{1} we compute $\rho _{ij}^{(t)}$ at iteration $t$. In this case, $\rho _{ij}^{(t)}$ represents the minimum required transmit power of device $i$ to reach an SINR of 1 while connecting to UAV $j$, given the fixed transmit power of other devices. In step \ref{2}, we find the minimum transmit power of device $i$ if it connects to the best UAV. Then, the index of the best UAV which is assigned to device $i$ is given in step \ref{3}. In step \ref{4} we update the transmit power of device $i$ in order to achieve an SINR of $\gamma$. Steps \ref{1} to \ref{4} must be repeated for all devices to obtain the optimal transmit power and the device association vectors. 
	

	\begin{algorithm} 
		\begin{small}
			\begin{MyColorPar}{black}
			\caption{\textcolor{black}{Iterative algorithm for joint power control and device-UAV association}}
			\label{Fixed}
			\begin{algorithmic}[1] 	
					
				\State \textbf{Inputs:} Locations of UAVs and IoT devices
				\State \textbf{Outputs:}  Device association vector ($\boldsymbol{c}$), and
				transmit power of all devices ($\boldsymbol{P}$).
				\State {Set $t=0$, and initialize  $\boldsymbol{P} ^{(0)}= \left( {P_1^{(0)},...,P_K^{(0)}} \right)$}.\label{0}
				
				\State {Define $\rho _{ij}^{(t)}=\frac{{{\sigma ^2} + \sum\limits_{k \in \mathcal{Z}_i} {P_k^{(t)}{{\bar g}_{kj}}} }}{{{{\bar g}_{ij}}}}$.} \label{1}
				
				\State {Compute  ${S_i}(\boldsymbol{P}^{(t)}) = \mathop {\min }\limits_{j \in \mathcal{K}} \rho _{ij}^{(t)}$.}\label{2}
				
				\State {Find  ${c_i}(\boldsymbol{P}^{(t)}) = \mathop {\arg \min }\limits_{j \in \mathcal{K}} \rho _{ij}^{(t)}$.}\label{3}

				\State {Update $P_i^{(t + 1)} = \min \left\{ {\gamma {S_i}(\boldsymbol{P}^{(t)}),{P_{\max }}} \right\},\,\forall i \in \mathcal{L}_n.$}\label{4}

				\State {Repeat steps \ref{1} to \ref{4} for all devices until $\boldsymbol{P}^{(t)}$ converges.}
				
				\State 	{$\boldsymbol{P}=\boldsymbol{P}^{(t)}$, $\boldsymbol{c}=\left[ {c_i}(\boldsymbol{P}^{(t)})\right], \forall i\in \mathcal{L}_n$}.

			\end{algorithmic}\vspace{-0.1cm}
			\end{MyColorPar}
		\end{small}
	\end{algorithm}
	\vspace{-0.3cm}
\end{MyColorPar}   
As shown in \cite{Yates}, after several iterations this algorithm quickly converges to the global optimal solution if the SINR of each device is equal to $\gamma$. Hence, by solving \textbf{(P1-a)}, we are able to find the optimal transmit power of the devices and the device association for any given fixed locations of the UAVs. 
Then, the device association and transmit power of the devices will be used as inputs for solving the second subproblem in which  the UAVs' locations need to be optimized (in Subsection III-B). 
 \subsubsection{Interference-free scenario}
At each update time $t_n$, if the number of active devices is lower than the number of orthogonal channels or equivalently $L_n\le R$, there will be no interference between the devices. Unlike in the interference scenario, here, the transmit power of each device can be computed only based on the channel gain between the device and its serving UAV. Therefore, considering (\ref{L_ave}), and (\ref{SINR}) without interference, the minimum transmit power of device $i$ in order to connect to UAV $j$ is $P_i=\gamma \sigma^2 \bar L_{ij}$. In this case, given the locations of the UAVs (fixed for all $ \boldsymbol{v}_j$), $\bar L_{ij}$ is known for all devices and problem \textbf{(P1-a)} can be simplified. Hence, the optimal association problem under minimum power in the interference-free scenario will be:\vspace{-0.1cm} 

\hspace{3cm}\textbf{(P1-b)}:\vspace{-0.4cm}
\begin{align} \label{assign}
&\min\limits_{A_{ij}} \sum\limits_{j = 1}^K {\sum\limits_{i = 1}^{L_n} {{A_{ij}}\bar{L}_{ij}} }, \\
{\rm{s}}{\rm{.t}}{\rm{.}}&\sum\limits_{j = 1}^K {{A_{ij}}}  = 1, \,\,\,\forall i\in\mathcal{L}_n,\\
&{A_{ij}}\bar{L}_{ij}\le \frac{P_\textrm{max}}{\gamma \sigma^2}, \,\,{A_{ij}} \in \{ 0,1\},\,\, \forall i\in\mathcal{L}_n, j\in\mathcal{K},\label{const2}
\end{align}
where $\bar{L}_{ij}$ is the average path loss between device $i$ and UAV $j$, which is known, give the locations of the UAV and the device. $A_{ij}$ is equal to 1 if device $i$ is assigned to UAV $j$, otherwise $A_{ij}$ will be equal to 0. Clearly, the optimization problem in (\ref{assign}) is an integer  linear programming (ILP). In general, this problem can be solved by using standard ILP solution methods such as the cutting plane. 
 However, these solutions might not be efficient as the size of the problem grows. In particular, due to the potentially high number of IoT devices, a more efficient technique for solving (\ref{assign}) is needed. Here, we transform problem (\ref{assign}) to a standard assignment problem \cite{BookAs} which can be solved in polynomial time. In the assignment problem, the objective is to find the optimal one-to-one assignment between two sets of nodes with a minimum cost. In our problem, the devices and the UAVs can be considered as two sets of nodes that need to be assigned to each other with an assignment cost of $L_{ij}$ between nodes $i$ and $j$. \textcolor{black}{However, compared to the classical assignment problem, \textbf{(P1-b)} has an additional constraint in (\ref{const2}) which results from the transmit maximum power constraint. This constraint indicates that device $i$ cannot be assigned to UAV $j$ if $\bar{L}_{ij}> \frac{P_\textrm{max}}{\gamma \sigma^2}$. Therefore, in the assignment problem we can consider $L_{ij}=+\infty$ to avoid assigning device $i$ to UAV $j$ when $\bar{L}_{ij}> \frac{P_\textrm{max}}{\gamma \sigma^2}$ that implies the constraint in (\ref{const2}) is violated. Subsequently, using the updated assignment costs, $L_{ij}$, problem \textbf{(P1-b)} will be transformed to the classical assignment problem which can be solved using the Hungarian method with a time complexity of  $O((L_nK)^3)$ \cite{ku}.} 
\textcolor{black}{We note that, in absence of interference, problems \textbf{(P1-a)} and \textbf{(P1-b)} have the same solution.} Next, we present the second subproblem of the original optimization problem (\ref{Assosiation1}) in order to optimize the UAVs' locations.\vspace{-0.3cm} 


\subsection{Optimal Locations of the UAVs}

In this section, given the optimal device association, our goal is to find the \textcolor{black}{sub-optimal} locations of the UAVs for which the total transmit power of the devices is minimized. In other words, considering the mobile nature of the UAVs, we intelligently update the location of each UAV based on the location of its associated IoT devices.    \vspace{-0.5cm}\\

\subsubsection{Interference scenario}
In this scenario, given the UAV-device associations, the optimization problem to find the 3D locations of the UAVs and the transmit power of the devices will be:

\hspace{2.5cm}\textbf{(P2-a)}:\vspace{-0.4cm}
\begin{align}
&\mathop {\min}\limits_{\boldsymbol{v}_j, \boldsymbol{P}} \sum\limits_{i=1}^{{L_n}} {{P_i}}\, ,\,\,\,\,\,\,\, \forall i \in \mathcal{L}_n, \,\, \forall j \in \mathcal{K},\label{AssosiationP2a}\\
\text{s.t.} \,\, &\frac{{{P_i}{{\bar g}_{ij}}(\boldsymbol{v}_j)}}{{\sum\limits_{k \in \mathcal{Z}_i} {{P_k}{\bar g}_{kj}(\boldsymbol{v}_j)}  + {\sigma ^2}}} \ge \gamma, \label{SINRP2a}\\
&0 < {P_i} \le {P_\textrm{max}}, \vspace{-0.3cm}\label{PmaxP2a}
\end{align}
where $\boldsymbol{v_j}=(x^\textrm{uav}_{j},y^\textrm{uav}_{j},h_j)$ indicates the 3D location of UAV $j$. Clearly, the channel gains used in (\ref{SINRP2a}) depend on the locations of the UAVs. Note that, according to (\ref{PLoS}) and (\ref{L_ave}), $\bar g_{ij}(\boldsymbol{v}_j)$ is a non-convex function of $\boldsymbol{v}_j$. Consequently, constraint (\ref{SINRP2a}) is also non-linear and non-convex. Furthermore, the transmit power of the devices and the UAVs' locations are mutually dependent. On the one hand, the location of each UAV must be determined such that its associated devices can connect to it with a minimum transmit power. On the other hand, the UAV's location will impact the amount of interference received from other interfering devices. Indeed, solving the optimization problem in \textbf{(P2-a)} is challenging as the problem is highly non-linear and non-convex. In particular, the complexity of this problem stems from the mutual dependence between the transmit power of the devices and the locations of the UAVs.  

Our proposed approach to solve \textbf{(P2-a)} is based on optimizing the location of each UAV separately. Note that, using the results of \textbf{(P1-a)}, for each UAV, the associated and non-associated devices and their transmit power, $P^*_i$, are known. Our proposed solution proceeds as follows. The cloud starts by considering a single UAV and then optimizing its location given the fixed transmit power for the non-associated devices. Then, the cloud updates the transmit power of the associated devices according to the new location of their serving UAV. Hence, at each step, the location of a UAV and the transmit power of its associated devices are updated. At each iteration, after finding $P^*_i$, we set $P_\textrm{max}=P^*_i$ for the next iteration. This ensures that the transmit power of the devices does not increase during the iterative process. The entire process is repeated by the cloud for all UAVs one-by-one, until the transmit power of the devices cannot be further reduced by changing the UAVs' locations.  Note that, at each step, one must determine the optimal location of each UAV such that the total transmit power of its associated devices is minimized. 
 
Now, let $\mathcal{C}_j$ be the set of
 devices' index associated to UAV $j$. Given (\ref{L_ave}), (\ref{AssosiationP2a}), and (\ref{SINRP2a}), the optimal location of UAV $j$ can be determined by solving the following problem:
 \begin{align}
& \mathop {\min }\limits_{\boldsymbol{v_j}} \sum\limits_{i\in \mathcal{C}_j} {F_i(\boldsymbol{v}_j)},\label{Opt} \\ 
\textrm{s.t.}\hspace{0.3cm}
&F_i(\boldsymbol{v}_j)=  {\gamma \left( {{\eta _1}P_\textrm{LoS}^{ij} + {\eta _2}P_\textrm{NLoS}^{ij}} \right){{\left( {{K_o}{d_{ij}}} \right)}^\alpha }\left[ {\sum\limits_{k \in {Z_i}} {\frac{{{P_k}}}{{\left( {{\eta _1}P_\textrm{LoS}^{kj} + {\eta _2}P_\textrm{NLoS}^{kj}} \right){{\left( {{K_o}{d_{kj}}} \right)}^\alpha }}} + {\sigma ^2}} } \right]},\label{Fj}\\
&{F_i(\boldsymbol{v}_j)}\le{P_i^*},\,\,\, \forall i \in \mathcal{C}_j,\label{Pi}
\end{align}
 Note that, $P_\textrm{LoS}^{ij}$, $P_\textrm{LoS}^{kj}$,  $d_{kj}$, and $d_{ij}$ depend on the locations of UAVs $(\boldsymbol{v_j})$. Also,  (\ref{Pi}) guarantees that the transmit power of each device is reduced by updating the location of serving UAV. 

Clearly, (\ref{Opt}), (\ref{Fj}), and (\ref{Pi}) are non-linear and non-convex. Considering the fact that the objective function and constraints are twice differentiable, we convert (\ref{Opt}) to a quadratic form which can be solved using efficient techniques. In particular, we adopt the sequential quadratic programming (SQP) method as one of the most powerful algorithms for solving large scale and constrained differentiable non-linear optimization problems \cite{boggs}. Clearly, considering the high non-linearity of (\ref{Fj}) as well as the large number of constraints, the SQP is a suitable method for solving our optimization problem. In the SQP method, the objective function is approximated by a quadratic function, and the constraints are linearized. Subsequently, the optimization problem is solved by solving multiple quadratic subproblems. \begin{MyColorPar}{black} In our optimization problem, to find the optimal location of UAV $j$, $\boldsymbol{{v}}_{j,k}$, we start with an initial point $\boldsymbol{v}_{j,k}$ (starting with $k=0$). Then, we use the first order necessary optimality or
Karush-Kuhn-Tucker (KKT) conditions to find the Lagrangian variables. In particular, we use:	
\begin{equation}
\nabla L(\boldsymbol{v}_{j,k},\boldsymbol{\lambda}_k ) = \sum\limits_{i \in \mathcal{C}_j} {\nabla{F_i}(\boldsymbol{{v}}_{j,k})}+\nabla \boldsymbol{w}_i(\boldsymbol{{v}}_{j,k})\boldsymbol{\lambda}_k=0, \label{Lag}
\end{equation}
where $L(\boldsymbol{v}_{j,k},\boldsymbol{\lambda}_k ) = \sum\limits_{i \in \mathcal{C}_j} {{F_i}(\boldsymbol{{v}}_{j,k})}  +\boldsymbol{\lambda}^T \boldsymbol{w}(\boldsymbol{{v}}_{j,k})$ is the Lagrangian function, $\boldsymbol{\lambda}_k$ is the vector of Lagrangian variables, and $\boldsymbol{w}(\boldsymbol{{v}}_{j,k})$ is a vector of functions with each element being $\boldsymbol{w}_i(\boldsymbol{{v}}_{j,k})={\left( {{F_i}(\boldsymbol{{v}}_{j,k}) - P_i^*} \right)}$.
Then, given $\boldsymbol{v}_{j,k}$, we determine the Lagrange variables by \cite{boggs}:	\vspace{0.01cm}
\begin{equation}
\boldsymbol{\lambda}_k=-\left[\boldsymbol{w}_i(\boldsymbol{{v}}_{j,k})^T \nabla \boldsymbol{w}_i(\boldsymbol{{v}}_{j,k})  \right]^{-1} \nabla \boldsymbol{w}_i(\boldsymbol{{v}}_{j,k})^T \sum\limits_{i \in \mathcal{C}_j} {\nabla{F_i}(\boldsymbol{{v}}_{j,k})}. \label{Lam}
\end{equation}
\end{MyColorPar}
In the next step, we update $\boldsymbol{v}_{j,k+1}=\boldsymbol{v}_{j,k}+\boldsymbol{d}_k$, where $\boldsymbol{d}_k$ is the solution to the following quadratic programming problem: \vspace{-0.1cm}
\begin{align}
&\boldsymbol{{d}}_k = \mathop {\arg \min }\limits_{\boldsymbol{d}} \sum\limits_{i \in \mathcal {C}_j} {\nabla {F_i}{{({\boldsymbol{v}_{j,k}})}^T}\boldsymbol{d}}  + \frac{1}{2}{\boldsymbol{d^T}}{\nabla ^2}\left[ {L({\boldsymbol{v}_{j,k}},{\boldsymbol{\lambda }_k})} \right]\boldsymbol{d},\label{QP}\\
\textrm{s.t.}\,\,\,&{F_i}({\boldsymbol{v}_{j,k}}) + \nabla {F_i}{({\boldsymbol{v}_{j,k}})^T}\boldsymbol{d} - P_i^* \le 0,\,\,\,\,\,\forall i \in \mathcal {C}_j, \label{dk}
\end{align}
where, $\nabla$ and ${\nabla ^2}$ indicate the gradient and Hessian operations. Clearly, (\ref{QP}) is an inequality constrained quadratic programming. Moreover, it can be shown that the Hessian matrix, ${\nabla ^2}\left[ {L({\boldsymbol{v}},{\boldsymbol{\lambda }_k})} \right]$, is not positive semidefinite, and, hence, (\ref{QP}) is non-convex in general. In this case, the two possible solution approaches are the active set, and the interior point methods. Typically, the active set method is preferred when the Hessian matrix is moderate/small and dense. The interior point, however, is a suitable approach when the Hessian matrix is large and sparse \cite{sch}. In our problem, due to the  potential possible high number of active devices, the number of constraints can be high. Therefore, the Hessian matrix, ${\nabla ^2}\left[ {L({\boldsymbol{v}_{j,k}},{\boldsymbol{\lambda }_k})} \right]$, is large and sparse, and, hence, the interior point method is used.

 Finally, based on (\ref{Lag})-(\ref{dk}), the \textcolor{blue}{sub-optimal} location of each UAV $(\boldsymbol{v}_j)$, given the fixed device association, will be determined. Next, we address the UAVs' location optimization in an interference-free scenario. \vspace{-0.5cm}\\  

\subsubsection{Interference-free scenario}
In the absence of interference, we are able to provide tractable analysis on the UAVs' locations optimization. 
Considering $\alpha=2$ for LoS ground-to-air propagation \cite{HouraniOptimal}, the optimal locations of the UAVs will be given by:  \vspace{1.7cm}

\hspace{3cm}\textbf{(P2-b)}:\vspace{-0.5cm}
\begin{align}\label{vj}
& \mathop {\min }\limits_{\boldsymbol{v_j}} \sum\limits_{i\in \mathcal{C}_j}K_o^2{\sigma ^2}\gamma \left( {{\eta _1}P_\textrm{LoS}^{ij} + {\eta _2}P_\textrm{NLoS}^{ij}} \right){{ {{d_{ij}}} }^2 }, \\  
\textrm{s.t.}\,\,&\left( {{\eta _1}P_\textrm{LoS}^{ij} + {\eta _2}P_\textrm{NLoS}^{ij}} \right)d_{ij}^2 \le \frac{{{P_{\max }}}}{{K_o^2{\sigma ^2}\gamma }},\,\,\, \forall i \in \mathcal{C}_j.\label{D}
\end{align}

This optimization problem is non-convex over $\boldsymbol{v_j}=(x_{j}^\textrm{uav},y_{j}^\textrm{uav},h_j)$. However, given any altitude $h_j$, we can provide a tractable solution to this problem. \begin{MyColorPar}{black} First, given $h_j$, we consider the following function that is used in (\ref{vj}): \vspace{-0.3cm}
\begin{equation}
q(d_{ij})=K_o^2{\sigma ^2}\gamma \left( {{\eta _1}P_\textrm{LoS}^{ij} + {\eta _2}P_\textrm{NLoS}^{ij}} \right)d_{ij}^2.\vspace{-0.1cm}
\end{equation}
Clearly, considering the fact that $0\le P_\textrm{LoS}^{ij} \le 1$, and $P_\textrm{NLoS}^{ij}=1-P_\textrm{LoS}^{ij}$, we have: 
\begin{equation}
K_o^2{\sigma ^2}\gamma \eta_1 d_{ij}^2\le q(d_{ij}) \le K_o^2{\sigma ^2}\gamma \eta_2 d_{ij}^2. \label{ineq}
\end{equation} 
From (\ref{ineq}), we can see that $q(d_{ij})$ is bounded between two quadratic functions that each is linearly proportional to $d_{ij}^2$. 
Now, using the least square estimation method, we find the coefficients $\alpha_1$ and $\alpha_2$ such that, given any $h_j$, $q(d_{ij})$ is approximated by the following convex quadratic function:\vspace{-0.2cm}  
\begin{equation}
q(d_{ij})\approx {\alpha _1}d_{ij}^2 + {\alpha _2}.\vspace{-0.2cm} \label{apr}
\end{equation}
where $\alpha_1$ and $\alpha_2$ are altitude dependent coefficients. Note that, using the quadratic approximation, the solution of (\ref{vj}) becomes more tractable.
\end{MyColorPar}

 Fig.\,\ref{Error} shows the error in the objective function (\ref{vj}) due the quadratic approximation. As we can see from Fig.\,\ref{Error} which is obtained based on the parameters in Table I, the error is less than 4$\%$ for different UAVs' altitudes. 

\begin{figure}[!t]
	\begin{center}
		\vspace{-0.2cm}
		\includegraphics[width=7.5cm]{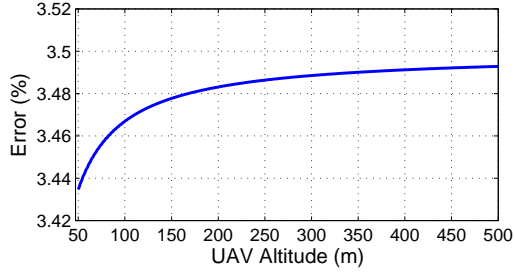}
		\vspace{-0.4cm}
		\caption{ \small Error in the objective function approximation.}\vspace{-1.3cm}
		\label{Error}
	\end{center}
\end{figure}  

Now, in  constraint (\ref{D}), we consider $D=\left( {{\eta _1}P_\textrm{LoS}^{ij} + {\eta _2}P_\textrm{NLoS}^{ij}} \right)d_{ij}^2$. Clearly, $D$ is an increasing function of $d_{ij}$ since ${\eta _1} - {\eta _2}\, $$<0$, and for a fixed altitude, LoS probability is a decreasing function of distance. Therefore, using $ {d_{ij}}^2 = {(x_i-x^\textrm{uav}_j)^2+(y_i-y^\textrm{uav}_j)^2+h_j^2 }$, and (\ref{apr}), for any given $h_j$ we can write the optimization problem (\ref{vj}) as:
\begin{align}
&\mathop {\min }\limits_{x^\textrm{uav}_j,y^\textrm{uav}_j} \sum\limits_{i \in {\mathcal{C}_j}} {{{(x^\textrm{uav}_j- {x_i})}^2} + {{(y^\textrm{uav}_j - {y_i})}^2}}+{h_j}^2 ,\label{theor}\\
\text{s.t.}\,\,&{(x^\textrm{uav}_j - {x_i})^2} + {(y^\textrm{uav}_j - {y_i})^2}+h_j^2 - \epsilon^2 \le 0, \,\,\, \forall \,\,\, i \in {\mathcal{C}_j},\, \text{and}\,\, j\in\mathcal{K}. \label{consupdate}
\end{align} 
where $\varepsilon  = \{ d|K_o^2{\sigma ^2}\gamma \left( {{\eta _1}P_\textrm{LoS} + {\eta _2}P_\textrm{NLoS}} \right)d^2 = {P_{\max }}\}$. Next, we derive the solution to problem (\ref{theor}) that corresponds to finding the \textcolor{black}{sub-optimal} UAVs' locations. 

\begin{theorem}
	\normalfont
	The solution to (\ref{theor}) is given by $\boldsymbol{s}^* = ({x^\textrm{uav}_j}^*,{y^\textrm{uav}_j}^*) =   - \boldsymbol{P{(\lambda )}}^{ - 1}\boldsymbol{Q(\lambda )}$, with the vector $\boldsymbol{\lambda}$ that maximizes the following concave function:
	\begin{align}
	&\mathop {\max {\rm{ }}}\limits_{\boldsymbol{\lambda}}  \frac{1}{2} \boldsymbol{Q(\lambda )}^T \boldsymbol{P{(\lambda )}}^{ - 1}\boldsymbol{Q(\lambda )}+   r(\boldsymbol{\lambda}),\\
	&\textnormal {s.t.} \,\,\,\boldsymbol{\lambda}  \ge 0,
	\end{align}
	where $\boldsymbol{P(\lambda)}= \boldsymbol{P}_o + \sum\limits_{i = 1}^{|{\mathcal{C}_j}|} {{\lambda _i}{P_i}}$, $\boldsymbol{Q(\lambda)}=\boldsymbol{Q}_o + \sum\limits_{i = 1}^{|{\mathcal{C}_j}|} {{\lambda_i}{Q_i}}$ and $r(\boldsymbol{\lambda})={r_o} + \sum\limits_{i = 1}^{|{\mathcal{C}_j}|} {{\lambda _i}{r_i}}$, with $\boldsymbol{P}_o$, $\boldsymbol{Q}_o$, $r_o$, $\boldsymbol{P}_i$, $\boldsymbol{Q}_i$, and $r_i$ given in the proof.	 	
\end{theorem}    
\begin{proof}
	As we can see from (\ref{theor}), the optimization problem is a  quadratically constrained quadratic program (QCQP) whose general form is given by \cite{boyd}:
	\begin{align}\label{QCQP}
	&\mathop {\min }\limits_{\boldsymbol{s}} \,{\rm{ }}\frac{1}{2}\boldsymbol{s}^T\boldsymbol{P}_o\boldsymbol{s} + \boldsymbol{Q}_o^T\boldsymbol{s} + {r_o},\\
	&\text{s.t.}\,\,\frac{1}{2}\boldsymbol{s}^T\boldsymbol{P}_i\boldsymbol{s} + \boldsymbol{Q}_i^T\boldsymbol{s} + {r_i},\,\,\,i \in {\mathcal{C}_j}.
	\end{align}
	Given (\ref{theor}) and (\ref{consupdate}), we have:\vspace{0.2cm}
	
	$\boldsymbol{{P}}_o = \left[ {\begin{array}{*{20}{c}}
		{2|{\mathcal{C}_j}|}&0\\
		0&{2|{\mathcal{C}_j}|}
		\end{array}} \right]$, $\boldsymbol{{P}}_i = \left[ {\begin{array}{*{20}{c}}
		2&0\\
		0&2 
		\end{array}} \right]$,  $\boldsymbol{Q}_o={\left[ {\begin{array}{*{20}{c}}
			{ - 2\sum\limits_{i = 1}^{|{\mathcal{C}_j}|} {{x_i}} }&{ - 2\sum\limits_{i = 1}^{|{\mathcal{C}_j}|} {{y_i}} }
			\end{array}} \right]^T}$, $\boldsymbol{{Q}}_i = {\left[ {\begin{array}{*{20}{c}}
			{ - 2{x_i}}&{ - 2{y_i}}
			\end{array}} \right]^T}$. Also, ${r_o} = \sum\limits_{i = 1}^{|{\mathcal{C}_j}|} {x_i^2}  + \sum\limits_{i = 1}^{|{\mathcal{C}_j}|} {y_i^2}$, and ${r_i} = x_i^2 + y_i^2+h_j^2-\epsilon^2$ with $\varepsilon  = \{ d|K_o^2{\sigma ^2}\gamma \left( {{\eta _1}P_\textrm{LoS} + {\eta _2}P_\textrm{NLoS}} \right)d^2 = {P_{\max }}\}$. Note that, $\boldsymbol{P}_o$ and $\boldsymbol{P}_i$ are positive semidefinite matrices, and, hence, the QCQP problem in (\ref{QCQP}) is convex. Now, we write the Lagrange dual function as:\vspace{-0.1cm}
	\begin{align}
	f(\boldsymbol{\lambda} ) &= \mathop \text{inf}\limits_{\boldsymbol s} \biggl[\frac{1}{2}\boldsymbol{s}^T\boldsymbol{P}_o\boldsymbol{s} + \boldsymbol{Q}_o^T\boldsymbol{s} + {r_o}+ \sum\limits_i {{\lambda _i}\left( {\frac{1}{2}\boldsymbol{s}^T\boldsymbol{P}_i\boldsymbol{s} + \boldsymbol{Q}_i^T\boldsymbol{s} + {r_i}} \right)}\biggr]\nonumber \\
	&= \mathop \text{inf}\limits_{\boldsymbol s} \left[ {\frac{1}{2}\boldsymbol{s}^T\boldsymbol{P(\lambda )s} + \boldsymbol{{Q{{(\lambda )}}}}^T\boldsymbol{s} + r(\boldsymbol{\lambda} )} \right].
	\end{align}  
	Clearly, by taking the gradient of the function inside the infimum with respect to $s$, we find $\boldsymbol{{s}}^* =  - \boldsymbol{P{(\lambda )}}^{ - 1}\boldsymbol{Q(\lambda )}$. As a result, using $\boldsymbol{{s^*}}$, $f(\boldsymbol{\lambda} ) = \frac{1}{2}\boldsymbol{Q}\boldsymbol{(\lambda )}^T\boldsymbol{P{(\lambda )}}^{ - 1}\boldsymbol{Q(\lambda )} + {r(\boldsymbol{\lambda} )}$. Finally, the dual of problem (\ref{QCQP}) or (\ref{theor}) will be:\vspace{-0.3cm}
	\begin{align}
	\text{max}\,\, f(\boldsymbol{\lambda}), \,\,\textnormal {s.t.} \,\,\,\boldsymbol{\lambda}  \ge 0,\vspace{-0.2cm}
	\end{align}
	which proves Theorem 1.
\end{proof}
\begin{MyColorPar}{black} Using Theorem 1, for a fixed altitude, we find the optimal 2D coordinates of the UAV, $\boldsymbol{s}^* =({x^\textrm{uav}_j}^*,{y^\textrm{uav}_j}^*)$. Then, the optimal UAV's altitude is the argument that minimizes the following one-dimensional function as:\vspace{-0.5cm}
\begin{equation}
h_j^* = \mathop {\arg \min }\limits_{{h_j}} \Big[ {{\alpha _1}\left( {h_j^2 + {{\left\| {\boldsymbol{s}^*} \right\|}^2}} \right) + {\alpha _2}} \Big]. \label{hj} \vspace{-0.1cm}
\end{equation}
where $\alpha_1$ and $\alpha_2$ are the altitude dependent coefficients given in (\ref{apr}). Note that, given (\ref{hj}), the \textcolor{black}{sub-optimal} altitude of the UAV is obtained via one dimensional search over a feasible range of altitudes. Consequently, we can determine the optimal 3D location of each UAV.
\end{MyColorPar}


Note that, the device association (presented in Section III), and  UAVs' locations optimization (in Section IV) are applied iteratively until there is no change in the  location update step. Clearly, at each iteration, the total transmit power of the devices is reduced and the objective function is monotonically decreasing. Hence, the solution converges after several iterations. Note that, our proposed approach provides a suboptimal solution to the original problem. Nevertheless, our solution has a reasonable accuracy but significantly fast compared to the global optimal solution that can be achieved by the brute-force search, as will be further corroborated in the simulations.  

Thus far, we considered the IoT network at one snapshot in the time duration $[0,T]$.  Next, we analyze the IoT network considering the entire time duration $[0,T]$ in which the set of active devices changes. In this case, to maintain the power-efficient and reliable uplink communications of the devices,  the UAVs must update their locations at different update times $t_n$.\vspace{-0.2cm}

\section{Update Times and Mobility of UAVs} 
Here, we find the optimal update time and trajectory of the UAVs to guarantee the reliable uplink transmissions of the IoT devices. Clearly, the trajectory of the UAVs, as well their update time depend on the activation process of the IoT devices. Furthermore, to move along the optimal trajectories, the UAVs must spend a minimum total energy on mobility so as to remain operational for a longer time. In the considered ground IoT network, the set of active IoT devices changes over time. Consequently, the UAVs must frequently update their locations accordingly. Note that, the UAVs do not continuously move as they must stop, serve the devices, and then update their locations. Furthermore, the mobility of the UAVs is also limited due to their energy constraints. Hence, the UAVs update their locations only at some specific times. In this case, during time interval $[0,T]$, we need to find update times $t_n$, $1 \le n \le N$ with $N$ updates, and a framework for optimizing the mobility of the UAVs at different update times. \textcolor{black}{Here, for tractability, we assume that the devices are synchronized at $t=0$. In this case, the synchronization process needs to be done only once during the entire activation period $[0,T]$. Note that, our optimization problems for  jointly finding the optimal UAVs' locations, the device association, and devices' transmit power at each update time do not depend on this synchronization assumption.\vspace{-0.3cm}}   

\subsection{Update Time Analysis}
First, we propose a framework to find the update times of the UAVs. 
As discussed in Section II, each UAV's trajectory consists of multiple stop locations (determined in update times) at which each UAV serves its associated ground devices. Clearly, the update times depend on the activation of the IoT devices during the given time period $[0,T]$. Indeed, the number of update times, $N$, impacts the optimal location and trajectory of the UAVs as well as the power consumption of the IoT devices. A higher number of updates leads to a shorter time interval between the consecutive updates. Hence, a lower number of devices will be active during the shorter time interval. In such a case, the active devices experience lower interference from each other while transmitting their data to the UAVs. Therefore, the IoT devices can use lower transmit power to meet their SINR constraint. However, a higher number of updates requires more mobility and higher energy consumption for the UAVs. Next, we provide insightful analysis on the update time based on the probabilistic and periodic activation models of the IoT devices.\vspace{0.01cm}

\subsubsection{Periodic IoT activation}  
In some applications such as weather monitoring, smart grids (e.g. smart meters), and home automation, the IoT devices can report their data periodically. Therefore, the devices are activated periodically. Let $\tau_i$ be the activation period of  device $i$ during  $[0,T]$.  
Without loss of generality, assume ${\tau _1} \le {\tau _2} \le ... \le {\tau _L}$. Due to the periodic nature of devices' activation, we can find the exact number of active devices at each update time $t_n$.\vspace{-0.1cm}

\begin{proposition} \normalfont{ 
		The exact number of active IoT devices at update time $t_n$ is given by:
		\begin{align}
		&{b_n} = \sum\limits_{i = 1}^L {\mathds{1}\left( {\left\lfloor {\frac{{{t_n}}^-}{{{\tau _i}}}} \right\rfloor  > \left\lfloor {\frac{{{t_{n - 1}}}}{{{\tau _i}}}} \right\rfloor } \right)}, \,\,\, n > 1,\\
		&{b_1} = \mathop {\arg \max }\limits_i \left\{ {{t_1} > {\tau _i}} \right\},
		\end{align}
		where $\mathds{1}(.)$ is the indicator function which can only be equal to 1 or 0, and ${t_n}^ -  = \mathop {\lim }\limits_{\varepsilon  \to {0^ + }} ({t_n} - \varepsilon )$.}
\end{proposition}

\begin{proof}
	User $i$ becomes active during $[t_{n-1},t_n)$ if there exists $q\in\mathds{N}$ such that ${t_{n - 1}} \le ~q{\tau _i} ~< ~{t_n}$. Thus, the number of activations of device $i$ before $t_n$ must be greater than the one until $t_{n-1}$. Considering the fact that the number of activations before $t_n$ is $\left\lfloor {\frac{{{t_n}^ - }}{{{\tau _i}}}} \right\rfloor$, and until $t_{n-1}$ is $\left\lfloor {\frac{{{t_{n - 1}}}}{{{\tau _i}}}} \right\rfloor$, we must have:\vspace{-0.2cm}
	\begin{equation}
	{\left\lfloor {\frac{{{t_n}^ - }}{{{\tau _i}}}} \right\rfloor  > \left\lfloor {\frac{{{t_{n - 1}}}}{{{\tau _i}}}} \right\rfloor}.
	\end{equation}
	Hence, the total number of active devices which need to be served at $t_n$ is equal to:
	\begin{equation}
	{b_n} = \sum\limits_{i = 1}^L {\mathds{1}\left( {\left\lfloor {\frac{{{t_n}}^-}{{{\tau _i}}}} \right\rfloor  > \left\lfloor {\frac{{{t_{n - 1}}}}{{{\tau _i}}}} \right\rfloor } \right)}.
	\end{equation} 	
	Finally, considering $t_0=0$, we can write $b_1$ as:	${b_1} = \mathop {\arg \max }\limits_i \left\{ {{t_1} > {\tau _i}} \right\}$. 
\end{proof}
Proposition 2 gives the exact number of devices that must be served by UAVs at each update time. In this case, the update times can be adjusted according to the number of devices that can be served by the UAVs. Indeed, knowing the exact number of active devices enables us to determine the update times in a deterministic and efficient way based on system requirements.

\subsubsection{Probabilistic IoT activation}  
Certain IoT devices can have probabilistic activations in applications such as health monitoring, and smart traffic control. In this case, each IoT device becomes active at time $t\in[0,T]$ following the beta distribution as given in (\ref{Beta}). For this scenario, we will next derive the specific update times as a function of the average number of active devices.\vspace{-0.1cm}

\begin{theorem}\textnormal{
	The update times during which, on the average, a total of $a_n$ devices must be served by the UAVs, are given by: \vspace{-0.2cm}
	\begin{align}
	&{t_n} = T \times {I^{ - 1}}\left( {\frac{{{a_n}}}{L} + {I_{\frac{{{t_{n - 1}}}}{T}}}\left( {\kappa  ,\omega } \right),\kappa ,\omega } \right),\,\, n>1, \vspace{0.3cm} \label{tn}\\
	&{t_1} = T \times {I^{ - 1}}\left( {\frac{{{a_1}}}{L},\kappa ,\omega } \right),
	\end{align}
	where $I_x(.)$ is the regularized incomplete beta function and ${I^{ - 1}}(.)$ is its inverse function. $L$ is the total number of IoT devices, and $[0,T]$ is the time interval \mbox{during which the devices can be active.}}\vspace{-0.2cm}
\end{theorem} 
\begin{proof}
First, we find the probability that each device becomes active in order to send its data to a UAV at update time $t_n$. As discussed in the system model, a device needs to transmit at time $t_n$ if it becomes active during time $[t_{n-1},t_n)$. Thus, the probability that each device needs to be served at $t_n$ is:
\begin{align}
{p_n} &= \int_{{t_{n - 1}}}^{{t_n}} {\frac{{{t^{\kappa  - 1}}{{(T - t)}^{\omega  - 1}}}}{{{T^{\kappa  + \omega  - 1}}B(\kappa ,\omega )}}\textrm{d}t = } \int_{\frac{{{t_{n - 1}}}}{T}}^{\frac{{{t_n}}}{T}} {\frac{{{t^{\kappa  - 1}}{{(1 - t)}^{\omega  - 1}}}}{{B(\kappa ,\omega )}}\textrm{d}t}, \nonumber\\
& =  \frac{{{B_{\frac{{{t_n}}}{T}}}(\kappa ,\omega ) - {B_{\frac{{{t_{n - 1}}}}{T}}}(\kappa ,\omega )}}{{B(\kappa ,\omega )}} = {I_{\frac{{{t_n}}}{T}}}(\kappa ,\omega ) - {I_{\frac{{{t_{n - 1}}}}{T}}}(\kappa ,\omega ),
\end{align}	
where ${B_x}(\kappa ,\omega ) = \int_0^x {{y^{\kappa  - 1}}{{(1 - y)}^{\omega  - 1}}{\textrm{d}y}}$ is the incomplete beta function with parameters $\kappa$ and $\omega$, and ${I_x}(.)$ is the regularized incomplete beta function.\\
Now, the average number of active devices at $t_n$ is given by:
\begin{align}
{a_n} &= \sum\limits_{k = 1}^L {\binom{L}{k}{p_n}^k{{(1 - {p_n})}^{L - k}}}  = L{p_n}\sum\limits_{k = 1}^L {\frac{{(L - 1)!}}{{(k - 1)!\left( {L - k} \right)!}}{p_n}^{k - 1}{{(1 - {p_n})}^{L - k}}},\nonumber \\
&=\sum\limits_{k' = 0}^{L'} {\frac{{(L')!}}{{(k')!\left( {L' - k'} \right)!}}{p_n}^{k' - 1}{{(1 - {p_n})}^{L' - k'}}}=Lp_n,\label{ave}
\end{align}
where in $(a)$, we used ${L'} = L - 1$ and ${k'} = k - 1$. Note that, (\ref{ave}) corresponds to the mean of a binomial distribution. Then we have:
\begin{equation}
L\left[ {{I_{\frac{{{t_n}}}{T}}}(\kappa ,\omega ) - {I_{\frac{{{t_{n - 1}}}}{T}}}(\kappa ,\omega )} \right] = {a_n}.
\end{equation}
This leads to:
\begin{equation}
{t_n} = T \times {I^{ - 1}}\left( {\frac{{{a_n}}}{L} + {I_{\frac{{{t_{n - 1}}}}{T}}}\left( {\kappa ,\omega } \right),\kappa ,\omega } \right).
\end{equation}
Finally, considering ${I_0}(.) = 0$, we find ${t_1} = T \times {I^{ - 1}}\left( {\frac{{{a_1}}}{L},\kappa ,\omega } \right)$.
\end{proof}
Clearly, the update times need to be determined based on the IoT devices' activation patterns. In fact, $t_n$ depends on the number of IoT devices, and their activation distribution. Furthermore, according to (\ref{tn}), each $t_n$ depends also on the previous update time, $t_{n-1}$. This is due to the fact that, the number of active devices that need to be served at $t_n$, depends on the update time difference $t_n-t_{n-1}$. Using Theorem 2, the update times of the UAVs can be set based on the average number of active devices. Typically, at each update time, the number of devices which need to be served by the UAVs should not be high in order to avoid high interference. However, considering the number of available resources (orthogonal channels and UAVs), it is preferable to serve a maximum number of active devices at each update time. Hence, in this case, the number of active devices at each update time must not be relatively low.
 Therefore, considering system requirements and different  parameters such as mutual interference between devices, acceptable delay for serving the devices, and number of available channels, an appropriate $t_n$ must be adopted. For instance, using Theorem 2, the update times can be set such that the average number of active devices be lower than the number of channels, $R$, to avoid interference between the devices. Next, we investigate the UAVs' mobility during the update times. \vspace{-0.3cm}  
\subsection{UAVs' Mobility}

Thus far, we have determined the update times as well as the stop locations at each update time. Here, we investigate how the UAVs should move between the stop locations at different update times. In this case, considering the energy limitation of the UAVs, $E_\textrm{max}$, we find the optimal trajectory of each UAV to guarantee reliable and energy-efficient uplink transmissions of active IoT devices. The UAVs update their locations according to the activity of the IoT devices. Therefore, the UAVs move from their initial locations at $t_{n-1}$ to a new optimal locations at $t_n$. This mobility should be done in such a way that the UAVs spend a minimum total energy on the mobility so as to remain operational for a longer time. In fact, given the optimal sets of UAVs' locations at  $t_{n-1}$ and $t_n$ obtained from Section III, we determine how to  move the UAVs between the initial and the new sets of locations in order to minimize total mobility of the UAVs. 

  Now, let $\mathcal{I}_{n-1}$ and $\mathcal{I}_n$ be two sets comprising the UAVs' locations at two consecutive update times $t_{n-1}$ and $t_n$. Our goal is to find the optimal mapping between these two sets in a way that the energy used for transportations (between two sets) is minimized. Not that, in our model, the total energy which each UAV can use for the mobility during $[0,T]$ is limited to $E_\textrm{max}$. Clearly, in the multiple updates (mobilities) during $[0,T]$, the maximum energy consumption of each UAV at each update is equal to the remaining energy of the UAV. Let $\Gamma_{n,k}$ be the remaining energy of the UAV at the location having index $k \in \mathcal{I}_{n-1}$  at time $t_n$. Then, we can write the following UAVs' mobility optimization problem: \vspace{-0.3cm}
  \begin{align} \label{transport1}
  &\min\limits_{\boldsymbol{Z}} \sum\limits_{l \in \mathcal{I}_{n}} {\sum\limits_{k \in \mathcal{I}_{n-1}} {{E_{kl}}{Z_{kl}}} }, \\
  \text{s.t.}\,&\sum\limits_{l \in \mathcal{I}_{n}} {{Z_{kl}}}  = 1,\,\, \sum\limits_{k \in \mathcal{I}_{n-1}} {{Z_{kl}}}  = 1,\\
  &{E_{lk}} \le {\Gamma _{n,k}} \,,\,\,\,{Z_{kl}} \in \{ 0,1\}, \,\, \forall k \in \mathcal{I}_{n-1}, \, \forall l \in \mathcal{I}_{n},\label{remain}
  \end{align}
  where $\mathcal{I}_{n-1}$ and $\mathcal{I}_{n}$, are the initial and new sets of UAVs' locations at times $t_{n-1}$ and $t_{n}$. $\boldsymbol{Z}$ is the $|\mathcal{I}_n|\times|\mathcal{I}_n|$ assignment matrix with each element $Z_{kl}$ being 1 if UAV $k$ is assigned to location $l$, and 0 otherwise. $E_{kl}$ is the energy used for moving a UAV from its initial  location with index $k \in \mathcal{I}_{n-1}$ to a new location with index $l \in \mathcal{I}_n$. Also, $\Gamma_{n,k}$ is the remaining energy for the UAVs at time $t_n$. Note that, (\ref{remain}) guarantees that UAVs remain operational until the end of the period $T$. 
	The total energy consumption of the rotary wing UAV while moving between two stop locations can be computed as done in \cite{di}:	
	\begin{equation}
	E = \frac{D}{v}\left( {{P_V} + {P_H}} \right),\label{E}
	\end{equation}	
	where $D$ is the distance between two stop locations, $D/v$ is the flight duration, $P_V$ is the power consumption for vertical movement, and $P_H$ is the power consumption for horizontal movement. Clearly, if the altitude difference between two stop locations is $\Delta h$, the effective vertical and horizontal velocities will be ${v_v} = v\sin \phi$ and ${v_h} = v\cos \phi$, with $\phi  = {\sin ^{ - 1}}\left( {\frac{{\Delta h}}{D}} \right)$.
\begin{MyColorPar}{black} 	
	According to \cite{zeng2016energy} and \cite{filippo}, $P_H$ is composed of parasitic power and induced power needed for overcoming the parasitic drag and the lift-induced drag. The parasitic power, based on \cite[equations (13.32), (13.27), and (11.3)] {filippo}, can be  given by:
	\begin{equation}
	{P_P} = \frac{1}{2}\rho {C_{D_o}}{A_e}{v_h^3} + \frac{\pi }{4}{N_b}c_b\rho {C_{{D_o}}}{\omega ^3}{R^4}\left( {1 + 3{{\left( {\frac{{{v_h}}}{{\omega R}}} \right)}^2}} \right), \label{PP}
	\end{equation}
	where $v_h$ is the effective horizontal velocity, $C_{D_o}$ is the drag coefficient, $\rho$ is the air density, $c_b$ is the blade chord, $N_b$ is the number of blades,  and $A_e$ is the reference area (frontal area of the UAV) \cite{zeng2016energy} and \cite{filippo}. We note that the second term in (\ref{PP}) represents the blade power profile.
	
	Using \cite [equations (13.19), (13.13), and (12.2)] {filippo}, the induced power (assuming zero tilt angle) can be computed by:\vspace{-0.3cm}
	\begin{equation}
	{P_I} = \omega RW \times \lambda, \label{PI}
	\end{equation}
	where $R$ is the rotor disk radius, $W$ is the weight of the UAV, and $\omega$ is the angular velocity. Also, given \cite[(13.18), (13.13), and (12.1)] {filippo}, we can find $\lambda$ by solving the following equation:
	\begin{equation}
	g\left( \lambda  \right) = 2\rho \pi \omega ^2 {R^4}\lambda \sqrt {\frac{{{v_h^2}}}{{{\omega ^2}{R^2}}} + {\lambda ^2}}  - W = 0. \label{lam}
	\end{equation}
	
	The power consumption due to the vertical climbing and descending (assuming rapid descent) can be given by 
	\cite[equations (12.35), (12.47), (12.50)]{filippo}:
	\begin{equation}
	{P_V} = 
	\begin{cases}\label{PV} 
	
	\frac{W}{2}{v_v} + \frac{W}{2}\sqrt {{v_v}^2 + \frac{{2W}}{{\rho \pi {R^2}}}},\,\,\, \textrm{climbing},	\\
	\frac{W}{2}{v_v} - \frac{W}{2}\sqrt {{v_v}^2 - \frac{{2W}}{{\rho \pi {R^2}}}},\,\,\, \textrm{\textrm{descending (in windmill state)}},
	\end{cases}	
	\end{equation}
	where $v_v$ is the effective vertical velocity.
	Finally, the total mobility energy consumption is computed using (\ref{E})-(\ref{PV}).
	\end{MyColorPar}
	
	 Clearly, the optimization problem in (\ref{transport1}) is an integer  linear programming (ILP). Following the similar approach we used for solving (\ref{assign}), we transform problem (\ref{transport1}) to a standard assignment problem which can be solved using the Hungarian method in a polynomial time with a complexity of $O(|\mathcal{I}_n|^3)$. To this end, we need to remove constraint (\ref{remain}) by considering $E_{lk}=+\infty$ when the constraint is not satisfied. To determine when (\ref{remain}) is not satisfied, we use $\mathcal{I}_{n-1}$ and $\mathcal{I}_{n}$ to compute $E_{lk}$, and compare it with the remaining energy of the UAVs, $\Gamma _{n,k}$. Then, in the objective function (\ref{transport1}), we replace each $E_{lk}$ corresponding to the unsatisfied constraint with $E_{lk}=+\infty$. Consequently, (\ref{transport1}) is transformed into a standard assignment problem. 
  The result of solving (\ref{transport1})  will be the assignment matrix, $\boldsymbol{Z}$, that optimally assigns the UAVs to the destinations. Therefore, the locations of the UAVs are updated according to the new destinations. ،Then, having the destinations of each UAV at different update times, \mbox{we can find the optimal trajectory of the UAVs.\vspace{-0.25cm}}  

\section{Simulation Results and Analysis}
For our simulations, the IoT devices are deployed within a geographical area of size $1\,\text{km}\times 1 \,\text{km}$. In this case, we consider a total number of 500 IoT devices which are uniformly distributed on the area. Furthermore, we consider UAV-based communications in an urban environment with $\psi=11.95$ and $\beta=0.14$ at 2\,GHz carrier frequency \cite{HouraniOptimal}. 
 Table I lists the simulation parameters. Here, we analyze the transmit power of the IoT devices, the energy consumption of UAVs on their mobility, and the update times. In our update time analysis, unless otherwise stated, we consider the probabilistic activation model for the IoT devices with the beta distribution  parameters $\kappa=3$, and $\omega=4$ \cite{3GPP}. When applicable, we compare our results with pre-deployed stationary aerial base stations (UAVs) scenario while adopting the optimal device association and power control technique of Subsection III-A. In the stationary case, the locations of UAVs are assumed to be fixed over the target area and they are not updated according to the devices' locations. All statistical results are averaged over a large number of independent runs.

\begin{table}[!t]
	\normalsize
	\begin{center}
		\caption{\small Simulation parameters.}
		\vspace{-0.4cm}
		\label{TableP}
		\resizebox{8cm}{!}{
			\begin{tabular}{|c|c|c|}
				\hline
				\textbf{Parameter} & \textbf{Description} & \textbf{Value} \\ \hline \hline
				$P_\textrm{max}$	&    Maximum transmit power of each device      &  $200\,\textrm{mW}$     \\ \hline
				$\alpha$	&     Path loss exponent for LoS links     &        2   \\ \hline
				$\sigma^2$	&     Noise power     &        -130\,dBm  \\ \hline
				$\gamma$	&     SINR threshold      &   5\,dB \\ \hline
				$L$	&     Total number of IoT devices     &   500 \\ \hline
				
				$\eta_1$	&     Additional path loss to free space for LoS     &   3\,dB \\ \hline
				
				$\eta_2$	&     Additional path loss to free space for NLoS      &   23\,dB \\ \hline

			\end{tabular}}
			
		\end{center}\vspace{-1.2cm}
	\end{table}


\begin{figure}[!t]
	\begin{center}
		\vspace{-0.2cm}
		\includegraphics[width=6.8cm]{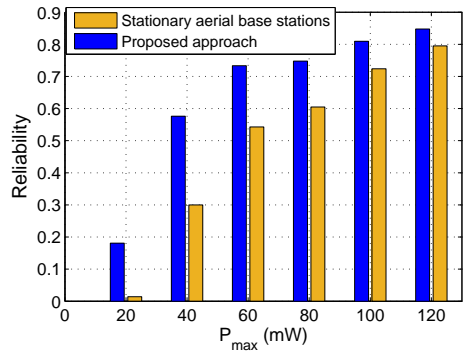}
		\vspace{-0.90cm}
		\caption{ \small Reliability comparison between the proposed approach and stationary aerial base stations using 5 UAVs.\vspace{-1.25cm}}
		\label{Reliability}
	\end{center}
\end{figure}

Note that, in the given IoT network, serving all the active devices may not be possible due to the limitations on the number of UAVs and the maximum transmit power of the devices. Thus, in Fig. \ref{Reliability}, we show the achieved system \textit{reliability} which, here, is defined as the probability that all the active devices can be served by the UAVs. Clearly, the reliability depends on the locations and transmit powers of the devices as well as the number of UAVs.

Fig. \ref{Reliability} shows the reliability as the maximum transmit power of the devices, $P_\textrm{max}$, varies. In this case, 5 UAVs are deployed to serve 100 active IoT devices. Clearly, as $P_\textrm{max}$ increases, the reliability also increases. In fact, for higher $P_\textrm{max}$ values, the devices have higher a chance to successfully connect to UAVs. From Fig. \ref{Reliability}, we can see that, our proposed approach leads to a significantly improved reliability compared to the case in which stationary aerial base stations are used. In particular, the difference between the reliability of the stationary case and our proposed approach is significant for lower $P_\textrm{max}$. Indeed, a higher reliability is achieved by dynamically optimizing the UAVs' locations based on the locations of the IoT devices. As shown in Fig. \ref{Reliability}, by increasing $P_\textrm{max}$ from 40\,mW to 100\,W, the reliability increases from 0.3 to 0.72 for the stationary case, while it increases from 0.58 to 0.82 in our proposed approach. Furthermore, the proposed approach yields a maximum of 28\% improvement in the system reliability. 

Fig. \ref{Locations} shows a snapshot of the UAVs' locations and their associated IoT devices (indicated by the same color) resulting from the proposed approach. In this figure, 5 UAVs are efficiently deployed to serve 100 active IoT devices  which are uniformly distributed on the area. In this case, all the devices are able to send their data to the associated UAVs by using a minimum total transmit power. Here, the 3D locations of the UAVs as well as the device association are determined based on the locations of the ground IoT devices and their transmit power. 

\begin{figure}[!t]
	\begin{center}
		\vspace{-0.2cm}
		\includegraphics[width=7.5cm]{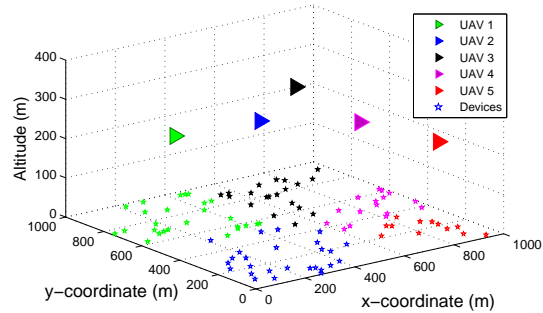}
		\vspace{-0.7cm}
		\caption{ \small UAVs' locations and associations for one illustrative snapshot.}\vspace{-1.1cm}
		\label{Locations}
	\end{center}
\end{figure}

In Fig. \ref{Pt_NumberofUAVs_new}, we show the total transmit power needed by the IoT devices for reliable uplink communications, versus the number of UAVs in the interference scenario. Clearly, the total transmit power of the IoT devices can be reduced by deploying more UAVs. For instance, considering 100 active devices and 20 available channels, using our proposed approach, the total transmit power decreases from 2.4\,W to 0.2\,W by increasing the number UAVs from 5 to 10. Furthermore, using the proposed approach, the total transmit power of the devices decreases by 45\% (on the average) compared to the stationary case. Clearly, for a lower number of UAVs, the proposed approach leads to higher power reduction compare to the stationary case. In other words, intelligently optimizing the locations of UAVs provides more power reduction gains when the number of UAVs is low. In fact,  for very dense networks with a high number of UAVs, updating the UAVs' locations is obviously no longer necessary compared to a case with a low number of UAVs. 
For instance, as we can see from Fig. \ref{Pt_NumberofUAVs_new}, the power reduction gain achieved by deploying 5 UAVs is around 7 times larger than the case with 10 UAVs.


\begin{figure}[!t]
	\begin{minipage}{.49\textwidth}
		\begin{center}
			\vspace{-0.2cm}
			\includegraphics[width=6.8cm]{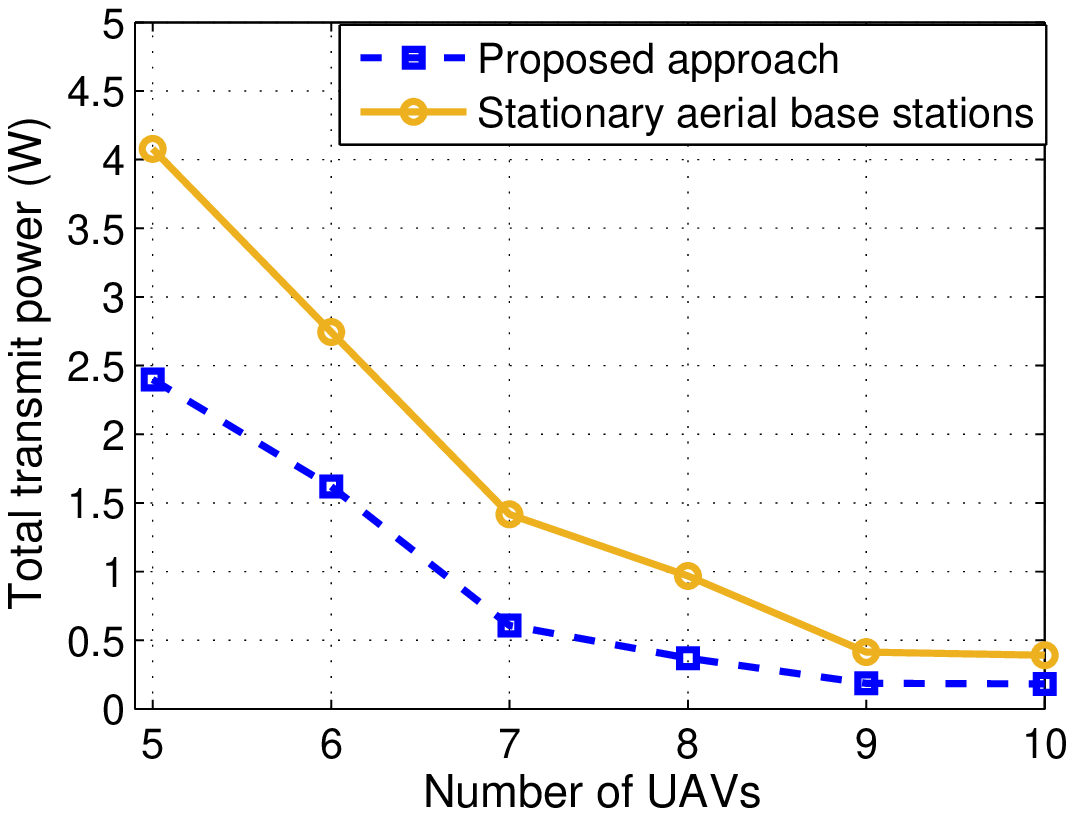}
			\vspace{-0.5cm}
			\caption{ \small Total transmit power of devices vs. number of\vspace{-0.1cm}\\ UAVs in the presence of interference.}\vspace{-1cm}
			\label{Pt_NumberofUAVs_new}
		\end{center}
	\end{minipage}
	\hfill
	\begin{minipage}{.49\textwidth}
			\begin{center}
				\vspace{-0.01cm}
				\includegraphics[width=6.8cm]{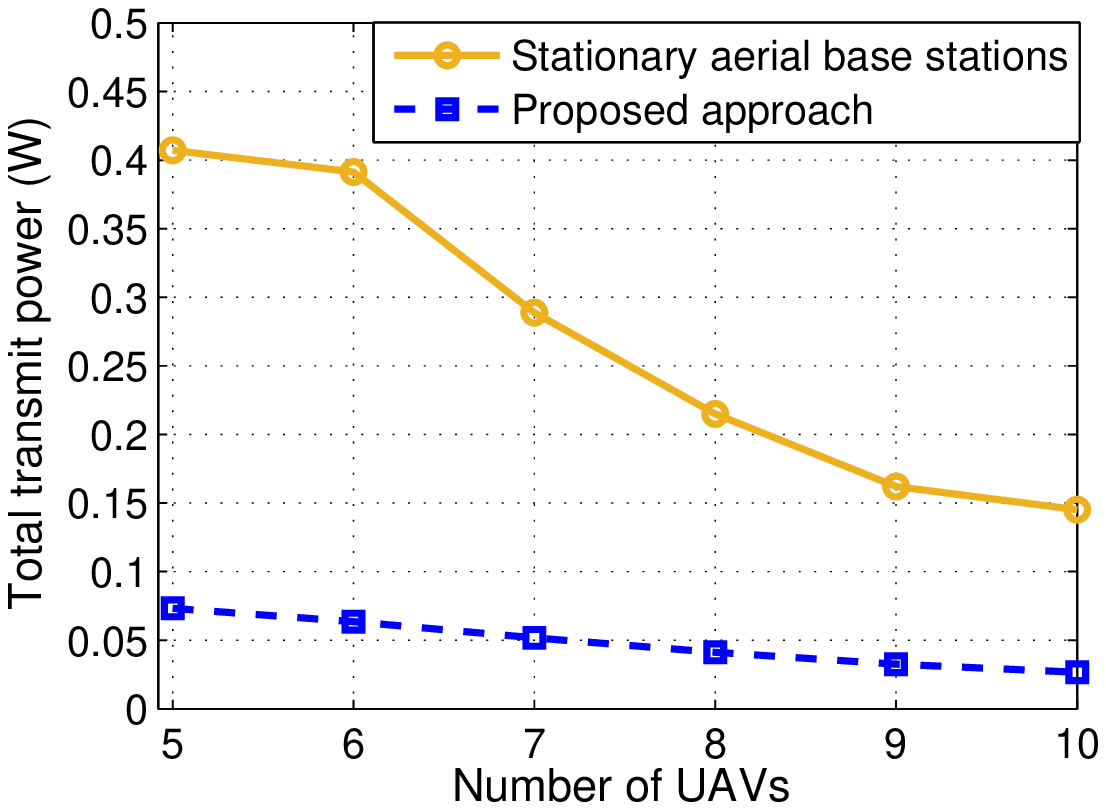}
				\vspace{-0.5cm}
				\caption{ \small Total transmit power of devices vs. number of \vspace{-0.1cm}\\UAVs in the interference-free scenario.}\vspace{-1cm}
				\label{NoInterference_Pt_vs_NumberofUAVs}
			\end{center}
	\end{minipage}
\end{figure}

Fig. \ref{NoInterference_Pt_vs_NumberofUAVs} shows the total transmit power of the IoT devices as a function of the number of UAVs in an interference-free scenario. Compared to the interference scenario, the devices can obviously use a lower transmit power for sending their data to the UAVs. For instance, by efficiently deploying only 5 UAVs, the devices can establish reliable uplink communications with a total transmit power of 70\,mW. Furthermore, Fig. \ref{NoInterference_Pt_vs_NumberofUAVs} shows that, our proposed approach leads to an average of 80\% power reduction compared to the stationary case.


Fig. \ref{Pt_resourceBlocks} shows the total transmit power of devices used for meeting the SINR requirement as the number of available channels varies. The result in Fig. \ref{Pt_resourceBlocks} corresponds to a case with 100 active devices which are served by 5 UAVs. Clearly, the total transmit power decreases as the number of channels increases. This is due to the fact that, when more orthogonal resources are available, the interference between the devices will decrease. As a result, each device can reduce its transmit power while connecting to the serving UAV. From Fig. \ref{Pt_resourceBlocks}, we can see that, by increasing the number of channels from 25 to 50, the total transmit power of devices can be reduced by 68\% in the proposed approach. In fact, the average number of interfering devices decreases from 4 to 2 when we increase the number of channels from 25 to 50. Consequently, less interference is generated by the devices while transmitting to the UAVs.  

\begin{figure}[!t]
	\begin{center}
		\vspace{-0.2cm}
		\includegraphics[width=6.8cm]{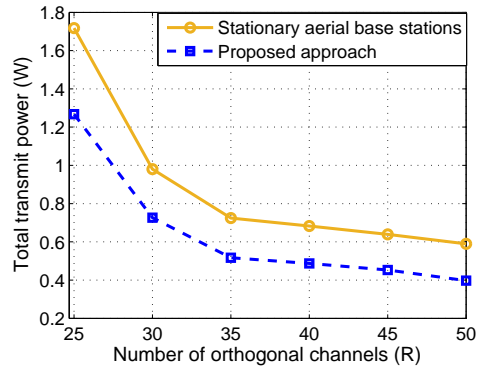}
		\vspace{-0.6cm}
		\caption{ \small Total transmit power of devices vs. number of orthogonal channels.}
		\label{Pt_resourceBlocks}\vspace{-1.4cm}
	\end{center}
\end{figure} 

\textcolor{black}{In Fig. \ref{NumberOfActiveDevices}, we show the average number of active devices that must be served by UAVs at different update times $t_n$ which are normalized by $T$.} \textcolor{black}{Clearly, the number of active devices at each update time depends the activation
process of the devices and the number of update times that indicates how frequently the UAVs
serve the devices. In Fig. \ref{NumberOfActiveDevices}, due to the beta distribution-based
activation pattern of the IoT devices, the number of active devices decreases when $t_n$ exceeds
0.5 for $N=10$.} 
 From Fig. \ref{NumberOfActiveDevices}, we can see that, for a higher number of update times or equivalently shorter time period between consecutive updates, the average number of devices that need to transmit their data decreases. For instance, considering $t_n=0.6$, the average number of active devices decreases from 180 to 80 when the number of updates increases from 5 to 10. We also note that, while a lower number of active devices leads to a lower interference between the devices, it requires more updates and mobility for the UAVs. Fig. \ref{NumberOfActiveDevices} also verifies that the analytical results in Theorem 2 match the simulations. Furthermore, in Fig. \ref{Periodic}, we show the exact number of active devices for the \emph{periodic activation case} obtained from Proposition 2. In this case, each device becomes active with a certain activation period, $\tau_i$. As expected, for a higher number of updates, a lower  number of active devices will need to be served by the UAVs. For instance, by increasing the number of updates from 10 to 30, on the average, the number of active devices decreases by 58\%. \textcolor{black}{Moreover, Fig. \ref{Periodic} shows that the maximum number of active devices for 10 updates is about two times larger than the case with 30 updates. Therefore, in order to avoid the interference between the devices, the number of orthogonal channels must be increased by a two-fold factor when the number of updates decreases from 30 to 10.} 
\begin{figure}[!t]
	\begin{center}
		\vspace{-0.2cm}
		\includegraphics[width=7.3cm]{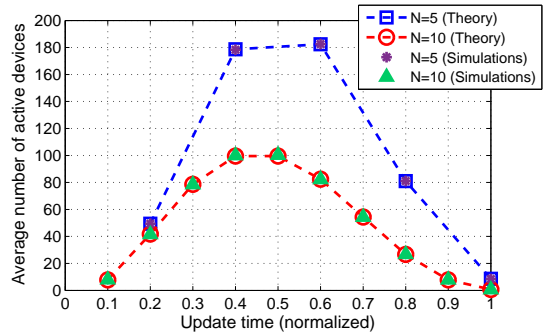}
		\vspace{-0.8cm}
		\caption{ \small Average number of active devices at update times for the probabilistic activation.\vspace{-1.0cm}}
		\label{NumberOfActiveDevices}
	\end{center}
\end{figure} 

\begin{figure}[!t]
	\begin{center}
		\vspace{-0.2cm}
		\includegraphics[width=6.4cm]{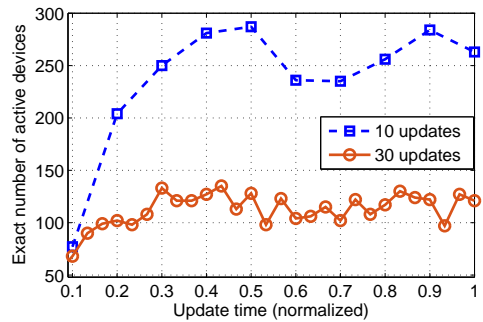}
		\vspace{-1.05cm}
		\caption{ \small Exact number of active devices at different update times for the periodic activation.\vspace{-1.3cm}}
		\label{Periodic}
	\end{center}
\end{figure} 

Fig. \ref{UpdateTime} presents a direct result of Theorem 2 that computes the update times based on the average number of active devices. Fig. \ref{UpdateTime} shows how to set update times in order to ensure that the number of devices (which needs to be served) at each update time does not exceed a specified number, $a$. As we can see from Fig. \ref{UpdateTime}, to achieve a lower value of $a$, updates must occur more frequently to reduce the time interval between the consecutive updates. For example, as can be seen from this figure, to meet $a= 100$, 75, and 50, the 5th update must occur at $t_n=0.41$, 0.55, and 1. Moreover, Fig. \ref{UpdateTime} shows that, the number of updates increases as $a$ decreases. For example, in this case, to reduce $a$ from 100 to 50, the number of updates needs to be doubled.

\begin{figure}[!t]
	\begin{center}
		\vspace{-0.2cm}
		\includegraphics[width=6.8cm]{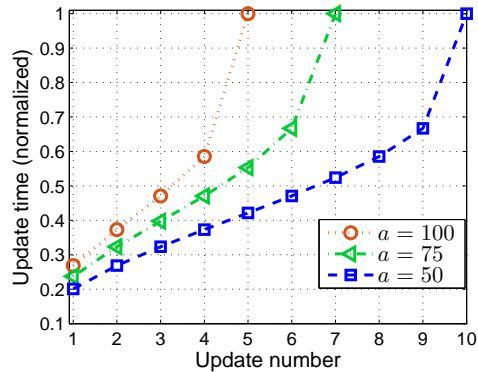}
		\vspace{-0.5cm}
		\caption{ \small Update times for different average number of active devices.}\vspace{-1.2cm}
		\label{UpdateTime}
	\end{center}
\end{figure}

%

\textcolor{black}{Fig. \ref{Energy_update} shows the impact of the number of updates on the amount of energy that the UAVs use to move. For our simulations, we have considered $v=10$\,m/s, $\rho=1.225$\,kg/$\textrm{m}^{-3}$, $\omega=20$\,rad/s, $R=0.5$\,m, $c_b=10$\,cm, $N_b=4$, and $W=50$\,N  \cite{filippo}. Intuitively, a higher number of updates requires more mobility of the UAVs. Therefore, by increasing the number of updates, the total energy consumption of the UAVs will also increase. As we can see from Fig.  \ref{Energy_update}, by increasing the number of updates from 3 to 6, the energy consumption of UAVs increases by a factor of 2.1 when the target area size is $1\,\text{km}\times 1 \,\text{km}$. Note that, the mobility of the UAVs also depends on the size of geographical area in which the devices are distributed. Clearly, on average, the UAVs need to move further for covering a larger area.} 

Interestingly, there is an inherent tradeoff between the number of updates, mobility of the UAVs, and transmit power of the IoT devices. In fact, considering Fig. \ref{Energy_update}, a higher number of updates leads to a higher energy consumption of the UAVs due to the higher mobility. In addition, as shown in Fig. \ref{NumberOfActiveDevices}, as the number of updates increases, a lower number of the IoT devices will be active at each update time, and, hence, there will be lower interference between the devices. As a result, the transmit power of the devices that is needed for satisfying the SINR requirement, can be reduced. Note that, as we showed in Fig. \ref{Pt_resourceBlocks}, the devices' transmit power decreases as the interference decreases (by increasing the number of orthogonal channels). Hence, while a higher number of updates leads to a lower devices' transmit power, it requires more UAVs' mobility.

\begin{figure}[!t]
	\begin{center}
		\vspace{-0.2cm}
		\includegraphics[width=6.5cm]{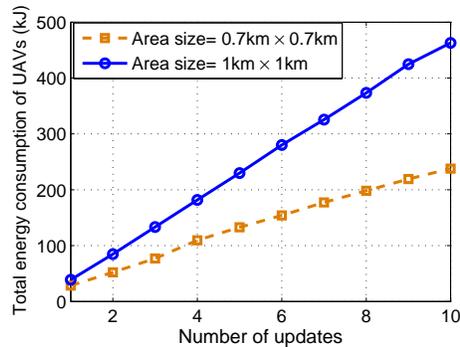}
		\vspace{-0.6cm}
		\caption{ \small Total UAV energy consumption vs. number of updates.}\vspace{-0.8cm}
		\label{Energy_update}
	\end{center}
\end{figure}


Fig. \ref{Convergence} shows the overall convergence of the proposed power minimization algorithm that is used for solving the original problem (\ref{Assosiation1}) considering 5 UAVs. As we can see from the figure, in this case, the total transmit power of the IoT devices converges after 5 iterations. In Fig. \ref{Convergence}, each iteration corresponds to a joint solution to the device association and UAVs' locations optimization problems. Clearly, after several iterations, updating the device association and UAVs' locations will no longer improve the solution.


In Fig. \ref{ComparedwithOptimal}, we show an example to compare the accuracy and time complexity of our proposed approach with the optimal solution obtained by an exhaustive search. \textcolor{black}{Here, to perform an exhaustive search over the continuous space, we have discretized the space with a resolution of 0.1\,m.} In this case, two UAVs are deployed to serve the devices. Clearly, the average gap between the proposed solution and the optimal solution is around 11\%. However, in this example, the proposed solution is around 500 times, on the average, faster than the optimal solution.\vspace{-0.34cm}

\begin{figure}[!t]
	\hspace{1cm}
 \begin{minipage}{.45\textwidth}
  \begin{center}
   \vspace{-0.2cm}
    \includegraphics[width=0.81\textwidth]{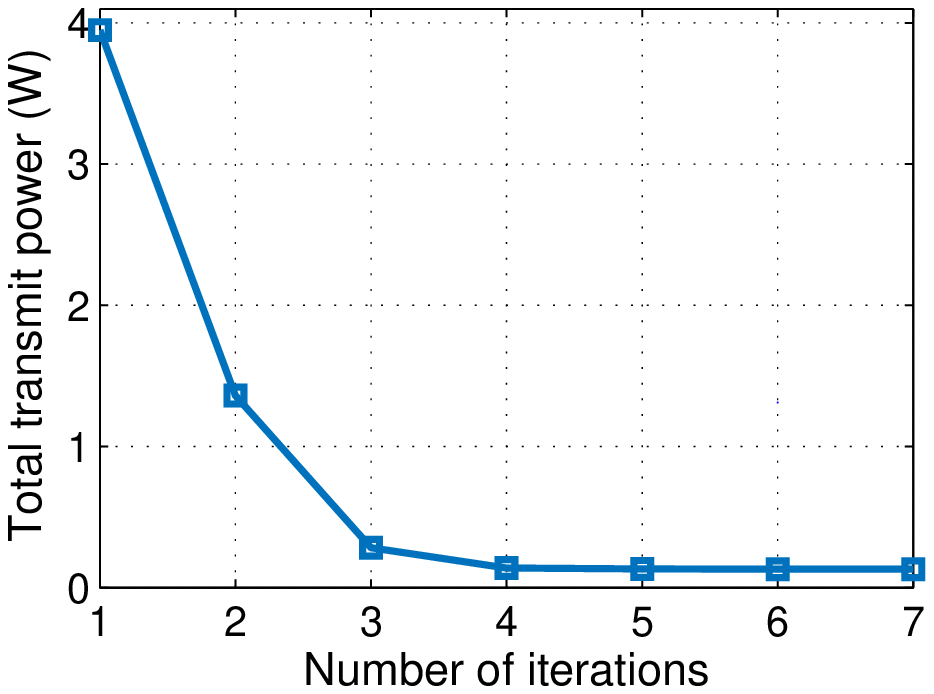}
    \vspace{-0.5cm}
    \caption{ \small Overall convergence of the algorithm.\vspace{0.1cm}}
    \label{Convergence}
  \end{center}\vspace{-1.41cm}
 \end{minipage}
\hspace{-0.cm}
 \begin{minipage}{.45\textwidth}
  \begin{center}
   \vspace{-0.4cm}
    \includegraphics[width=0.81\textwidth]{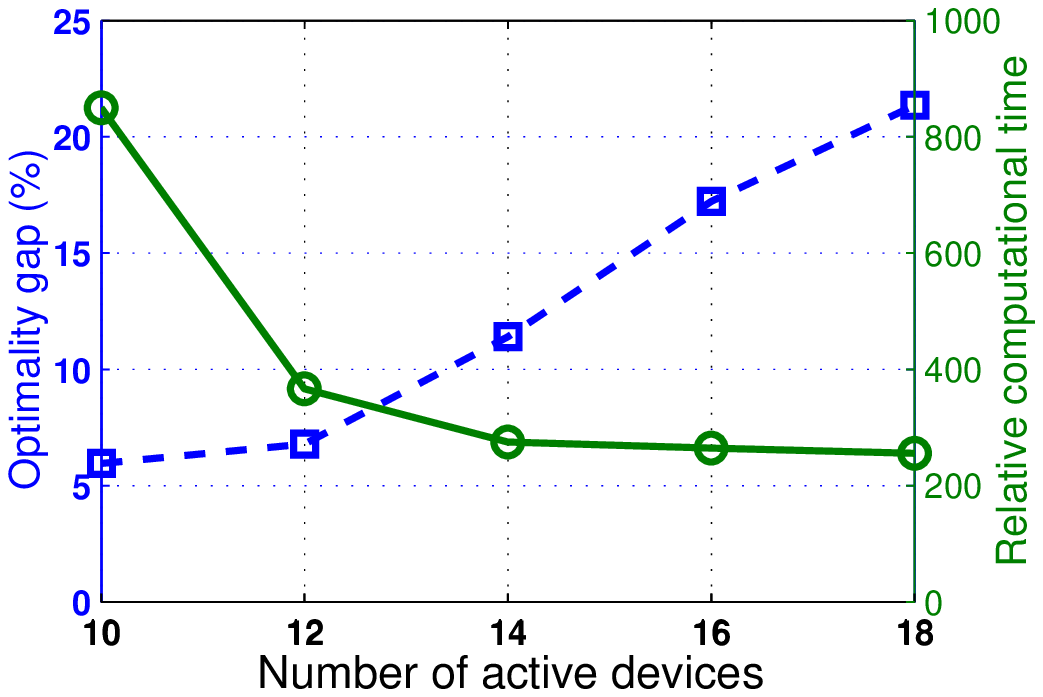}
    \vspace{-0.5cm}
    \caption{ \small Proposed approach vs. optimal solution.\vspace{-0.2cm}}
    \label{ComparedwithOptimal}
  \end{center}\vspace{-1.41cm}
   \end{minipage}
\end{figure}

\section{Conclusion}\vspace{-0.2cm}
In this paper, we have proposed a novel framework for efficiently deploying and moving UAVs to collect data in the uplink from ground IoT devices. In particular, we have determined the jointly optimal UAVs' locations, device association, and uplink power control of the IoT devices such that the total transmit power of the devices under their SINR constraints is minimized. In addition, we have investigated the effective movement of the UAVs to collect the IoT data in a time-varying IoT network. For this case, based on the devices’ activation process, we have derived the update time instances at which the
UAVs must update their locations. 
Furthermore, we have obtained the optimal trajectories that are used by the UAVs to dynamically serve the  IoT devices with a minimum energy consumption.
 The results have shown that by intelligently moving and deploying the UAVs, the total transmit power of the devices significantly decreases compared to the case with pre-deployed stationary aerial base stations. Moreover, there is a fundamental tradeoff between the number of updates, the UAVs' mobility, and the devices' transmit power. \vspace{-0.1cm} 
  
\def\baselinestretch{1.5}
\bibliographystyle{IEEEtran}

\bibliography{references}
\vspace{-0.4cm}
\end{document}